\documentclass[conference,compsoc]{IEEEtran}
\IEEEoverridecommandlockouts
\ifCLASSOPTIONcompsoc
  \usepackage[nocompress]{cite}
\else
  \usepackage{cite}
\fi
\ifCLASSINFOpdf
\else
\fi

\usepackage{amsmath}
\usepackage{amsthm}
\usepackage{amssymb}
\usepackage{mathtools}
\usepackage{bbm}
\usepackage{enumitem}
\usepackage{hyperref}
\usepackage{physics}
\usepackage[mathscr]{euscript}
\usepackage{eurosym}
\usepackage{istgame}
\usepackage{cancel}
\usepackage{mathtools}
\usepackage{caption}
\usepackage{subcaption}
\usepackage{breakcites}
\usepackage{footmisc}
\usepackage{cuted}
\usepackage{todonotes}

\usepackage[ruled]{algorithm2e} 
\DeclarePairedDelimiter{\ceil}{\lceil}{\rceil}
\DeclarePairedDelimiter{\floor}{\lfloor}{\rfloor}

\SetAlFnt{\small}
\SetAlCapFnt{\small}
\SetAlCapNameFnt{\small}
\SetAlCapHSkip{0pt}
\IncMargin{-\parindent}

\usepackage{etoolbox}

\newtoggle{fullversion}
\toggletrue{fullversion}  

\newcommand{\FullOnly}[1]{%
  \iftoggle{fullversion}{#1}{}%
}

\newcommand{\CameraOnly}[1]{%
  \iftoggle{fullversion}{}{#1}%
}

\newcommand{\R}{\mathbb{R}}
\newcommand{\N}{\mathbb{N}}
\newcommand{\No}{\mathcal{N}}
\newcommand{\X}{\mathscr{X}}
\newcommand{\D}{\mathscr{D}}

\newcommand{\B}{\mathcal{B}}

\newcommand{\1}{\mathbbm{1}}
\newcommand{\Z}{\mathbb{Z}}

\newcommand{\Y}{\mathscr{Y}}
\newcommand{\F}{\mathcal{C}}

\newcommand{\G}{\mathcal{A}}
\newcommand{\Ga}{\mathcal{G}}
\newcommand{\htlc}{\mathcal{H}}

\newcommand{\Q}{\mathcal{Q}}
\newcommand{\Po}{\mathscr{P}}

\newcommand{\gvec}{\mathbf{g}}

\newcommand{\vvec}{v}

\newcommand{\lambdavec}{\boldsymbol\lambda}
\renewcommand{\P}{\mathcal{P}}

\newcommand{\bgame}{\B}
\newcommand{\bmodel}{\mathscr{O}}

\newcommand{\bg}{\textup{bg}}
\newcommand{\ag}{\textup{ag}}
\newcommand{\nog}{\textup{ng}}
\newcommand{\cg}{\textup{cg}}
\newcommand{\Uup}{\textup{U}}
\newcommand{\Gup}{\textup{G}}
\newcommand{\Lup}{\textup{L}}

\CameraOnly{
    \newif\ifshowtheoremlinks
    \showtheoremlinksfalse   
    
    \newtheoremstyle{theoremwithlink} 
      {3pt}  
      {3pt}  
      {\itshape} 
      {}     
      {\bfseries} 
      {.}    
      { }    
      {%
        \bfseries #1~#2%
        \ifshowtheoremlinks
          \ (#3)%
        \fi
      }
      
    \theoremstyle{theoremwithlink}
}
\newtheorem{theorem}{Theorem}[section]

\theoremstyle{plain}
\newtheorem{remark}[theorem]{Remark}
\newtheorem{definition}[theorem]{Definition}
\newtheorem{lemma}[theorem]{Lemma}
\newtheorem{corollary}[theorem]{Corollary}

\newcommand{\tx}[2]{x_{#1}^{#2}}

\newcommand{\pim}[1]{}
\newcommand{\gf}[1]{}
\newcommand{\za}[1]{}
\newcommand{\fr}[1]{}
\newcommand{\mm}[1]{}

\def\bitcoin{%
  \leavevmode
  \vtop{\offinterlineskip 
    \setbox0=\hbox{B}%
    \setbox2=\hbox to\wd0{\hfil\hskip-.03em
    \vrule height .3ex width .15ex\hskip .08em
    \vrule height .3ex width .15ex\hfil}
    \vbox{\copy2\box0}\box2}}

\begin{document}
%
\title{A Composable Game-Theoretic Framework for Blockchains\thanks{This research was partially funded by the European Research Council (ERC) under the European Union’s Horizon 2020 research (grant agreement 101141432-BlockSec), by the Austrian Science Fund (FWF) 10.55776/F85 (subprojects F85-10, F85-12, and F85-03), and by the Vienna Science and Technology Fund (WWTF) through the projects 10.47379/ICT22045 and 10.47379/ICT25056.}}

\makeatletter
\newcommand{\linebreakand}{%
  \end{@IEEEauthorhalign}
  \hfill\mbox{}\par
  \mbox{}\hfill\begin{@IEEEauthorhalign}
}
\makeatother


\FullOnly{

    
    
    \author{
        \IEEEauthorblockN{Zeta Avarikioti}
        \IEEEauthorblockA{TU Wien, Common Prefix\\georgia.avarikioti@tuwien.ac.at}
        \and
        \IEEEauthorblockN{Georg Fuchsbauer}
        \IEEEauthorblockA{TU Wien\\georg.fuchsbauer@tuwien.ac.at}
        \and
        \IEEEauthorblockN{Pim Keer}
        \IEEEauthorblockA{TU Wien\\pim.keer@tuwien.ac.at}
        \linebreakand
        \IEEEauthorblockN{Matteo Maffei}
        \IEEEauthorblockA{TU Wien\\matteo.maffei@tuwien.ac.at}
        \and
        \IEEEauthorblockN{Fabian Regen}
        \IEEEauthorblockA{TU Wien\\fabian.regen@tuwien.ac.at}
    }
}


%


\maketitle

\begin{abstract}
    Blockchains rely on economic incentives to ensure secure and decentralised operation, making incentive compatibility a core design concern. However, protocols are rarely deployed in isolation. Applications interact with the underlying consensus and network layers, and multiple protocols may run concurrently on the same chain. These interactions give rise to complex incentive dynamics that traditional, isolated analyses often fail to capture. 
    
    We propose the first compositional game-theoretic framework for blockchain protocols. Our model represents blockchain protocols as interacting games across the application, network, and consensus layers. It enables formal reasoning about incentive compatibility under composition by introducing two key abstractions: the cross-layer game, which models how strategies in one layer influence others, and cross-application composition, which captures how application protocols interact concurrently through shared infrastructure. 

    \CameraOnly{
        We illustrate our framework through case studies on Hashed Timelock Contracts (HTLCs) and Layer-2 protocols
        showing how compositional analysis reveals new subtle incentive vulnerabilities and supports modular security proofs. Also, by introduction of a novel rational miner model, we derive new conditions for the robustness of timelocks to bribing attacks. 
    }
    \FullOnly{
        We illustrate our framework through case studies on Hashed Timelock Contracts (HTLCs), Layer-2 protocols, and Maximal Extractable Value (MEV) showing how compositional analysis reveals new subtle incentive vulnerabilities and supports modular security proofs. Also, by introduction of a novel rational miner model, we derive new conditions for the robustness of timelocks to bribing attacks.
    }
    
\end{abstract}


%
\IEEEpeerreviewmaketitle

\section{Introduction}
\label{sec:introduction}

Blockchains currently secure more than {\$3 trillion} in digital assets. As most operate in an open, permissionless fashion, such as {Bitcoin} and {Ethereum}, their security fundamentally relies on financial incentives. That is, participants invest resources, whether computational power or cryptocurrency holdings, in exchange for rewards such as block rewards and transaction fees. 

This incentive-driven security model supports a broader ecosystem commonly referred to as Web3---a decentralised digital environment where applications operate on distributed networks rather than being governed by centralised entities. Web3 is composed of \emph{multiple interdependent layers}, each fulfilling distinct roles. At its core, the \emph{blockchain (or consensus) layer} is responsible for ordering transactions in a trustless manner. On top of it operates the \emph{application layer}, which includes decentralised applications (dApps) and Layer 2 (L2) solutions designed to enhance functionality, efficiency, and scalability of the underlying blockchain layer. Underpinning these layers is the \emph{network layer}, which facilitates secure and reliable communication among geographically distributed participants.

These layers \emph{do not operate in isolation}. Their interdependence gives rise to complex \emph{security dependencies} that fall outside the scope of traditional game-theoretic models, which typically focus on a single layer while abstracting away the rest. In particular, \emph{incentives leak across layers}, introducing subtle but critical vulnerabilities.
Notable examples of such cross-layer issues include \emph{timelock bribing attacks}~\cite{nadahalli2021timelocked, avarikioti2022suborn} and \emph{Maximal Extractable Value (MEV) exploits}~\cite{daian2020flash}. In timelock bribing attacks, adversaries at the application layer {bribe} validators to {censor transactions}, allowing them to extract funds from protocols relying on time-based smart contracts. MEV exploits, in contrast, arise from validators (or external actors) observing pending transactions and manipulating their inclusion on-chain, for instance via \emph{frontrunning}, \emph{backrunning}, or \emph{sandwich attacks}~\cite{zhou2023sokdefi}. While timelock bribing constitutes an explicit {security breach}, MEV exploits reflect an emergent consequence of \emph{misaligned incentives between the network and consensus layers}.
Both cases underscore the limitations of layer-specific security models and show the need for a \emph{composable game-theoretic security framework} that captures the incentive dynamics spanning multiple layers and supports the design of economically secure blockchain protocols.

{Nevertheless, existing analyses largely fail to address this need.} 
As we argue in Section~\ref{subsec:related-work}, prior works either \emph{do not account for rational behaviour}, remaining rooted in cryptographic models with only honest and Byzantine participants, e.g.,~\cite{garay2024bitcoin,kiayias2017ouroboros,yossi2017algorand,phil2019snowwhite,maofan2019hotstuff}, \emph{are too general to produce meaningful results in the blockchain setting}~\cite{ghani2016compositional}, or \emph{adopt monolithic rational models}, abstracting away cross-layer interactions, e.g., \cite{zappala2021framework,badertscher2018but,kiayias2016blockchain,aiyer2005bar,pass2017fruitchains, avarikioti2020cerberus,avarikioti2023muskateer}. As a result, none of these approaches offer a \emph{general, composable framework that can model cross-layer incentives or the concurrent execution of multiple applications on a shared blockchain infrastructure.}


\subsection{Our Contributions}  

This work addresses this gap by introducing a composable game-theoretic framework for blockchains that enables formal reasoning about protocol security across multiple layers of the blockchain stack. To this end, we characterise\FullOnly{ and formalise} the distinct layers of a blockchain ecosystem---the application, network,
and consensus layers---alongside the interfaces that govern their interactions.

We first define the \emph{blockchain game}, which models the strategic behaviour of miners or validators in response to a set of fee-paying transactions. We then define the \emph{application game}, which captures the actions of protocol participants operating an application-layer protocol and whose outcome is a set of transaction triples (each consisting of a transaction, a fee, and a timestamp) intended for submission to the blockchain. 
\CameraOnly{In the full version, we moreover introduce the \emph{network game}}\FullOnly{Finally, we introduce the \emph{network game}}, a novel intermediary layer that models how these transactions are disseminated to the consensus layer. This is the layer where phenomena such as Maximal Extractable Value (MEV) emerge, as it governs the visibility and timing of transactions across the system.\CameraOnly{ In this work, we restrict ourselves to a simple network model that allows us to omit the network game to simplify the presentation.}

To reason about incentive security in such settings, we introduce a rigorous notion of \emph{game-theoretically secure protocols}. Informally, we require that the expected utility of protocol participants does not increase when the protocol is executed in a composed setting---interacting with the blockchain game\FullOnly{, the network game,} and other application games---compared to when it is executed in isolation.
At the core of our framework are two new constructs: the \emph{cross-layer game}, which integrates the application\FullOnly{, network,} and blockchain layers to capture their strategic \emph{sequential} interplay, and \emph{cross-application composition}, which enables reasoning about the concurrent execution of multiple protocols on a shared blockchain infrastructure.


We demonstrate our framework's applicability through several examples, focusing mostly on timelock bribing as an approachable example of cross-layer incentives. 
\gf{..and because it does not require formalizing the network layer(?)}
Here, a party tries to censor a transaction by bribing miners to include a conflicting timelocked transaction once it becomes valid. This leads to the \emph{Censor-Only Blockchain Game}, a novel model where rational miners/validators can only censor, not fork or withhold blocks. To the best of our knowledge, it is the first to explicitly account for a \emph{myopic} portion of hashrate or stake not susceptible to bribes. For Bitcoin, these myopic miners may be small miners that together hold a significant share of hashrate. We find that with an estimated $2\%$ myopic hashrate, a timelock of $T=144$ blocks resists censoring if the bribe is below $36$ times the attacked transaction's fee.

Our Censor-Only Blockchain Game also supports various leader selection mechanisms, from unpredictable protocols like Bitcoin, where the leader is unknown until block discovery, to globally predictable ones like Ethereum, with a public leader schedule. This enables studying the vulnerability of applications such as payment channels and optimistic rollups to timelock bribing on blockchains with different leader selection mechanisms. We derive concrete conditions under which these applications are less vulnerable on unpredictable chains than on globally predictable ones.

\CameraOnly{
    Moving on, we showcase the expressivity of our framework through concrete blockchain protocols. First, we describe the composition of Hashed Timelock Contracts (HTLCs) across multiple payment channels (PCs), identifying those parameter regimes under which compositional security holds in the presence of rational participants. The framework's modularity becomes apparent here, as we can apply our abstract analysis of Censor-Only Blockchain Games in multiple different settings. In the full version, we illustrate how adding a network model can, when composed with an application game, lead to rational miners deviating from consensus to exploit MEV opportunities. This showcases how our framework can surface emergent vulnerabilities that layer-specific analyses fail to capture.
}
\FullOnly{
    Moving on, we showcase the expressivity of our framework through concrete blockchain protocols. First, we describe the composition of Hashed Timelock Contracts (HTLCs) across multiple payment channels (PCs), identifying those parameter regimes under which compositional security holds in the presence of rational participants. The framework's modularity becomes apparent here, as we can apply our abstract analysis of Censor-Only Blockchain Games in multiple different settings. Second, we illustrate how even a network model can, when composed with an application game, lead to rational miners deviating from consensus to exploit MEV opportunities. This showcases how our framework can surface emergent vulnerabilities that layer-specific analyses fail to capture.
}

Beyond these examples, we outline a broader class of settings that our framework naturally captures. These include interactions between decentralized exchanges and oracles, cross-chain protocols such as atomic swaps, and multi-layer incentive mechanisms like proposer-builder separation (PBS). While not formalised as case studies, these examples underscore the versatility and broader applicability of our framework for reasoning about complex, concurrent protocols within modern blockchain ecosystems.

\subsubsection*{\textbf{Summary of Contributions}} 
\begin{itemize}[leftmargin=*]
    \CameraOnly{
        \item We develop the first composable game-theoretic framework for blockchain ecosystems, capturing incentive interactions across application and consensus layers.
        To the best of our knowledge, our framework represents the first modular approach to reason about blockchain protocol security with rational participants, enabling the formulation of security guarantees for abstract classes of protocols. 
    }
    \FullOnly{
        \item We develop the first composable game-theoretic framework for blockchain ecosystems, capturing incentive interactions across application, network, and consensus layers. To the best of our knowledge, it represents the first modular approach to reason about blockchain protocol security with rational participants, enabling the formulation of security guarantees for abstract classes of protocols. 
    }
    \CameraOnly{
        \item We apply our framework in multiple settings (timelock bribing, HTLC composition), showcasing its ability to capture incentive dynamics and expose vulnerabilities (such as the wormhole attack \cite{malavolta2018anonymous}). We further outline broader applications, covering decentralised exchanges, atomic swaps, and proposer-builder separation, to demonstrate the framework’s potential to model complex settings spanning multiple applications, layers, and chains.
    }
    \FullOnly{
        \item We apply our framework in multiple settings (timelock bribing, HTLC composition, MEV), showcasing its ability to capture incentive dynamics and expose vulnerabilities (such as the wormhole attack \cite{malavolta2018anonymous}). We further outline broader applications, covering decentralised exchanges, atomic swaps, and proposer-builder separation, to demonstrate the framework’s potential to model complex settings spanning multiple applications, layers, and chains.
    }
    \item As part of our case studies, we develop a new family of blockchain models, accounting for myopic hashrate and the blockchain's leader selection mechanism. This framework allows us to give concrete timelock lower bounds to prevent bribing attacks for a given share of myopic hashrate. Also, we can derive conditions under which unpredictable chains are less vulnerable to timelock bribing than globally predictable chains. 
\end{itemize}

\subsubsection*{\textbf{Paper Organisation}}  
In Section~\ref{sec:model}, we introduce our framework by translating protocols into formal games, allowing us to define incentive compatibility in a composable setting. Section~\ref{sec:total-composition} then formalises the composition of these games and identifies conditions under which incentive compatibility is preserved. To illustrate the applicability of our framework, Section~\ref{sec:examples} presents several representative case studies. This allows a better comparison to the existing literature in Section~\ref{subsec:related-work}. Finally, Section~\ref{sec:conclusion} concludes and outlines directions for future research. 

\section{Model}
\label{sec:model}

\FullOnly{
    Our framework aims to support a broad class of protocols operating atop a common blockchain infrastructure. We consider three distinct layers: a \emph{network layer}, composed of nodes that communicate over a peer-to-peer network; a \emph{consensus layer}, which establishes an ordered sequence of transactions; and an \emph{application layer}, which captures the internal logic of decentralised applications. Our framework models each of these layers separately and formalises their interfaces. 
}

\CameraOnly{
    Our framework aims to support a broad class of protocols operating atop a common blockchain infrastructure. We consider three distinct layers: a \emph{consensus layer}, which establishes an ordered sequence of transactions;  an \emph{application layer}, which captures the internal logic of decentralised applications; and a \emph{network layer}, 
    composed of nodes that communicate over a peer-to-peer network. Our framework models each of these layers separately and formalises their~interfaces. 

    To ease the presentation of the framework, however, the formal modelling of the network layer will be deferred to the full version. As the interfaces between layers are formalised, it will be straightforward to insert later on a network layer between the application and consensus layers. In this paper, we assume a synchronous network model with instantaneous message delivery and a global clock, as is common in prior work~\cite{garay2024bitcoin,kiayias2017ouroboros}.
}

The framework enables us to reason about the composition of the different layers and consequently study the \emph{game-theoretic} security properties of protocols operating within a blockchain ecosystem in a modular way.
\begin{figure}[h]
    \centering
    \includegraphics[width=0.7\linewidth]{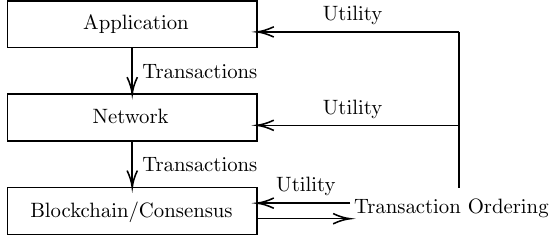}
    \caption{High-level overview of the blockchain layers.}
    \label{fig:layers}
\end{figure}

Game-theoretic security concerns the behaviour of \emph{rational players} who act to maximise their utility, typically measured in cryptocurrency holdings. A protocol is said to be \emph{incentive-compatible (IC)} if following its prescribed behaviour yields the highest expected utility for all participants, making deviation unprofitable.

Participants behave rationally if they act to maximise their utility. The utility function can be defined to model different behaviours. For instance, participants may always follow the protocol, which can be modelled by assigning infinite utility when they comply and zero otherwise. Alternatively, adversarial participants may derive utility from minimising another participant’s utility. In the following, we make the natural assumption that everyone tries to maximise their cryptocurrency holdings, but we stress that our framework can also accommodate other utility functions.

We assume that participants can collude, meaning they can coordinate strategies and aggregate their utilities. They effectively merge into a single entity capable of executing joint actions. They can also freely share any information they possess. Conversely, non-colluding participants only communicate through the blockchain. We formally define collusion and information-sharing in Section~\ref{sec:total-composition}.

\begin{remark}
\label{rem:ic}
    Incentive compatibility does not assess the quality of a protocol’s design. A flawed protocol---such as one that destroys half of all participants' funds---could still be IC if rational participants adhere to its prescribed actions. While such protocols may be impractical, our focus is solely on whether rational agents will follow the protocol {\em once they choose to participate.}
\end{remark}

The cryptocurrency holdings of all players (which they try to maximise) are determined by the \emph{state} of the blockchain. This state can be updated by players sending transactions over the network.
The role of the blockchain is to create an ordered sequence of those transactions. We will denote by $\X$ the set of all transactions that could be made within the blockchain ecosystem, and denote by $\Omega$ the set of all orderings of subsets of transactions of $\X$. Typically, a blockchain creates a transaction ordering by batching together transactions into blocks. Such a block is in itself a sequence of transactions. Therefore, an ordering of blocks always defines an ordering of transactions. That is, if we denote by $\Omega'$ the set of all orderings of blocks in $\X$, there exists a straightforward, surjective mapping from $\Omega'$ to $\Omega$. We abstract the behaviour of turning a set of transactions into a block ordering into the following mapping.
\begin{definition}
\label{def:blockchain}
    We define a \emph{block behaviour $\beta:\Po(\X)\to\Omega'$} as a function which has as input a set $X\subseteq\X$ of transactions, and outputs a sequence $(A_h)_{h\geq0}\in\Omega'$ of blocks of transactions, where for each $h\geq0$, the \emph{block} $A_h=(x_{h,k})_{k=1}^{K_h}$ is a \emph{finite}\pim{please check} sequence of transactions, such that $\bigcup_{h\geq0}\qty{x_{h,k}|1\leq k\leq K_h}=X$, i.e., such that \emph{exactly} the transactions in $X$ are ordered. We call a specific instance of the output of $\beta$ a \emph{block ordering}. The block ordering $(A_h)_{h\geq0}$ defines a transaction ordering, or \emph{blockchain ordering} in $\Omega$ by concatenation: $A_0\|A_1\|\ldots$. 
\end{definition}
Our definition of a block behaviour is as abstract as possible, in order to accommodate for the widest range of existing blockchain systems. In particular, it allows transactions to appear more than once or to conflict with other transactions (for example, for UTXO chains, by spending the same coins). The transaction ordering therefore still needs to be interpreted according to a certain set of rules (e.g., only the first occurrence of a transaction counts) to determine how the state of the blockchain actually changes. Since we are interested in making game-theoretic statements, we abstract the specific ruleset away by defining an \emph{execution function} $\omega_i:\Omega\to\R$, which for a given ordering returns the balance of a blockchain participant $i$. 

\subsection{The Blockchain Game}
\label{subsec:blockchain-game}
We are interested in turning a block behaviour into an object, which we can reason with game-theoretically. After all, the ordering that comes out upon input of a set of transactions is potentially determined by decisions made by rational actors, who try to maximise a given utility. These actors could be miners in a Proof-of-Work blockchain, or validators in a Proof-of-Stake blockchain. 

To this end, we introduce the \emph{blockchain game}. This is not a game in the traditional game-theoretic sense due to the lack of a utility function. To keep our framework as modular as possible, we refrain from defining the utility function as part of this blockchain game and introduce it only in Section~\ref{subsec:framework}. 
The players in this game are the miners or validators. We will use these terms interchangeably. We denote the player set by $M$. Upon input of a set of transactions, their decisions will lead to a certain blockchain ordering. In our framework, we explicitly mention two attributes alongside each transaction: a posting time and a transaction fee function. A block behaviour may use these attributes to determine the ordering it outputs, but note that explicitly keeping track of time and fee does not contradict Definition~\ref{def:blockchain}. These attributes are considered implicitly present in the transaction, such that a block behaviour can use this information when given a set of transactions. 

Intuitively, the posting time determines when a miner/validator learns about the existence of a transaction, and the transaction fee function determines, via the execution function, what utility the miners/validators receive if the transaction is included. By explicitly stating the posting time and transaction fee function, we extract from the transaction the information that players in the blockchain game need to make decisions. We therefore speak of \emph{transaction triples}. These are tuples $(x,t,f)$, with $x$ a transaction, $t\in\N_0$ the posting time of that transaction, and $f:\Omega\to\D(\R_{\geq0}^{|M|})$ the transaction fee function mapping each blockchain ordering to a probability distribution over the space of fee assignments to the miners. For a set $X$ of transactions, we denote by $\Y(X)$ the set of all sets of transaction triples of $X$, i.e., 
$\Y(X)=\Po\big(\big\{(x,t,f):x\in X,t\in\N_0,f\in(\Omega\to\D(\R_{\geq0}^{|M|}))\big\}\big)$.

The blockchain game thus takes as input a set of transaction triples $Y$ for some set of transactions $X$. The players in this game then all decide on a strategy, which results in an ordering $b\in\bmodel_M(Y)$ of the transactions present in $Y$. 
We denote here by $\bmodel_M(Y)\subseteq\Omega$ the set of all possible orderings from an input set $Y$ of transaction triples for a given set $M$ of blockchain game players. Recall that $\Omega$ is the set of all possible orderings. More correctly, given $Y$, the strategies of the players give a strategy profile, which results in a probability measure on $\bmodel_M(Y)$, as there might still be randomness involved in determining the ordering. 
\begin{definition}[Blockchain Game]
\label{def:b-game}
    Let $Y$ be a set of transactions triples.
    The \emph{blockchain game $\B(Y)$ induced by $Y$} is a tuple $(M,Y,\Sigma_{\bg},\pi_{\bg})$, where:
    \begin{itemize}[leftmargin=*]
        \item $M$ is the set of miners/validators, which are the \emph{players} of the blockchain game, together with their respective hashrates (for miners) or stakes (for validators).
        \item $\Sigma_{\bg}$ is the set of strategy profiles. For any $(x,t,f)\in Y$, players of the blockchain game only become aware of the transaction $x$ and its fee function $f$ at time $t$.
        \item $\pi_{\bg}:\Sigma_{\bg}\to\D(\bmodel_M(Y))$ is the \emph{ordering function}, mapping each $\sigma_{\bg}\in\Sigma_{\bg}$ to a probability measure on $\bmodel_M(Y)$.
    \end{itemize}
\end{definition}
\begin{remark}
\label{rem:constant-fee}
    We will often use the shorthand $(\tx{}{},t,f)$, where $f\in\R_{\geq0}$, to indicate that a fee of $f$ is awarded to whichever miner includes the transaction $\tx{}{}$, which would be the standard Bitcoin fee mechanism. 
\end{remark}

Capturing the full complexity of a block behaviour in a game model is far from trivial. We therefore consider three variants of a simplified block behaviour where miners (or \emph{validators}) choose which transactions to include and when, but always build on the same chain and never withhold valid blocks, as they would in selfish mining \cite{eyal2018majority}. We assume only the provided transactions pay fees, with no other sources of miner income like block rewards. We refer to this as a \emph{Censor-Only} block behaviour. Depending on whether validators know in advance which round(s) they will mine, we define three blockchain games:

\begin{itemize}[leftmargin=*]
    \item Unpredictable: Just as in Bitcoin, no miner knows whether they mine a block, until they actually mine it. There is no way to predict who will mine the next block. 
    \item Globally predictable: Everyone knows for a given number of blocks in the future who will mine each block. An example of this is Ethereum.
    \item Locally predictable: Each miner knows for a given number of blocks in the future which blocks are assigned to it, but is not aware of which miners get to mine the other blocks. Examples are Algorand and Cardano.
\end{itemize}
Due to our framework's modularity, we can easily swap out one variant for the other, or any variant for a more realistic model adhering to the blockchain game format. In the rest of this paper, we work with the following formal definition of the three Censor-Only variants.
\begin{definition}[Censor-Only Blockchain Games]
\label{def:nfnwmg}
    We define three types of Censor-Only Blockchain Games (COBGs). These games are all induced by a set $Y$ of transaction triples and all define a set $M=\qty{0,\ldots,m}$ of miners/validators. 
    
    \smallskip\noindent\textbf{Gameplay.} The games proceed in rounds indexed by $t$, starting from round $t=0$. At the start of round $s\geq0$, all miners learn about the transaction triples which have posting time $s$. At the end of round $s$, a block $b_s$ is added to the blockchain ordering based on the miner strategies in that round, which vary with the COBG type. All miners become aware of the block $b_s$ before the start of the next round $s+1$. 
    
    \smallskip\noindent\textbf{Hashrate distribution.} For all COBGs, the set of miners is represented as a \emph{hashrate distribution} (or \emph{stake distribution} in PoS) $\lambdavec=\qty(\lambda_j)_{j=0}^m$, with $\lambda_m\geq\ldots\geq\lambda_1>0$ the hashrates of $m$ rational miners, and $\lambda_0\geq0$ such that $\sum_{j=0}^m\lambda_j=1$. We call $\lambda_0$ the \emph{myopic} hashrate, aggregating the hashrate of each myopic miner, who always aims to build a block which gives maximal utility to that miner in that round (once an execution function is defined).
    
    \smallskip\noindent\textbf{Strategies.} The three types of COBGs differ in the strategies that players can take. These strategies implicitly define the strategy profile sets and ordering functions, for which we omit exact definitions. 
    \begin{itemize}[leftmargin=*]
        \item The \emph{Unpredictable COBG (U-COBG) $\bgame_{\Uup}(Y)$ induced by $Y$} is a tuple $(\lambdavec,Y,\Sigma_{\Uup,\bg},\pi_{\Uup,\bg})$. For round $s$, each miner $j$ will create a block proposal $b_{j,s}$. At the end of the round, one of the $m+1$\footnote{Because all myopic miners follow the same strategy, their block proposal will be the same.} block proposals will be randomly selected according to the hashrate distribution $\lambdavec$.
        \item The Globally Predictable (L-COBG) and Locally Predictable (G-COBG) COBGs $\bgame_{\Gup}(Y)$ and $\bgame_{\Lup}(Y)$ induced by $Y$ with period $D$ are tuples $(\lambdavec,Y,D,\Sigma_{\Gup/\Lup,\bg},\pi_{\Gup/\Lup,\bg})$. In these games, at the start of each \emph{epoch} $e$, i.e., the set of rounds $eD,eD+1,\ldots,(e+1)D-1$, with $D$ the \emph{epoch period}, a \emph{validator assignment} $\vvec^e$ is drawn at random from $\qty{0,\ldots,m}^D$, where miner $j$ has probability $\lambda_j$ of being selected for each round in the epoch to mine the block in that round. In that case, $v^e_{[s\bmod D]}=j$.
        \begin{itemize}
            \item In G-COBG, $\vvec^e$ is known to everyone. 
            \item In L-COBG, each miner $j$ only learns $\vvec^e|_j$, which are the values $t\in\Z_D$ for which $v_t^e=j$, i.e., when $j$ gets to mine a block.
        \end{itemize}
        The strategy of miner $j\in M\setminus\qty{0}$ is the sequence of functions $(\mu_{j,t})_{t\geq0}$, which, given $Y$, maps for each $t\geq0$, $\vvec^{\floor{t/D}}$ (in G-COBG) or $\vvec^{\floor{t/D}}|_j$ (in L-COBG) and $(b_s)_{s=0}^{t-1}$ to a block $b_t$ if $v^{\floor{t/D}}_{[t\bmod D]}=j$, and to $\perp$ otherwise. 
    \end{itemize}
\end{definition}
\begin{remark}[Myopia]
\label{rem:lambda0}
    Including the myopic hashrate makes COBGs more general than traditional rational miner models (e.g., \cite{aumayr2024securing}), which typically consider only a finite set of miners (i.e., the special case $\lambda_0=0$). In a permissionless system like Bitcoin, this neglects small, potentially unknown miners outside this set who, unlikely to mine again, include all known non-conflicting transactions maximising immediate fees. Ignoring this permissionless nature may lead to conclusions like miners censoring transactions indefinitely. To account for this, U-COBGs with a value $\lambda_0>0$ capture the non-negligible share of hashrate from negligibly-sized non-colluding rational miners that effectively behave \emph{as one honest miner}. This is because they all follow the same myopic strategy, which leads to the same block proposal under the implicit assumption these miners receive the same transactions. We stress that one can still explicitly define miners with (negligibly) small hashrate as separate, rational entities as players in the set $\qty{1,\ldots,m}$. In G-COBGs and L-COBGs, this intuition does not hold, because even small miners may censor if selected in an epoch. In these cases, we simply interpret $\lambda_0$ as a portion of hashrate/stake that behaves myopically for unknown reasons. 
\end{remark}
\begin{remark}
\label{rem:time}
    \pim{please check}
    In the Censor-Only blockchain game, by introducing the notion of rounds indexed by the posting time, we aligned the posting time with the blockchain block numbers. This is a design choice of the COBG, but by no means required for the framework. Indeed, we may for example express the posting time in seconds, according to some clock, and have blocks that come at random intervals according to this clock. In the remainder of this paper, aligning posting time and block heights is only done in the context of COBGs. 
\end{remark}
\subsection{The Application Game}
\label{subsec:application-game}
The transaction triples that will serve as the input of the blockchain game will come from the specific application that we are studying. Our goal is now to encapsulate everything that can happen while running some application protocol, into a game-like structure. Studying this structure will allow us to determine whether a protocol $\Pi$ is secure, i.e., incentivises participants to follow the intended protocol behaviour. The idea is that the internal protocol logic, and how the application outputs are submitted to the blockchain, should be represented by a strategy profile, which is then mapped to a set of transaction triples that serves as input to the blockchain game.
Our framework will be able to reason about any protocol $\Pi$ that can be written as a tuple $(N,\G,\overline{\sigma}_\ag)$, with $N$ the set of protocol participants, $\G$ a \emph{parametrised application game}, and $\overline{\sigma}_\ag$ the intended protocol behaviour(s) (IPB). 
\begin{definition}[Application Game]
\label{def:l2-game}
    An \emph{application game} $\G$ is a tuple $(N,X,\Sigma_\ag,\pi_\ag)$, where:
    \begin{itemize}[leftmargin=*]
        \item $N$ is the set of $n=|N|$ application protocol participants, referred to as \emph{(protocol) players}.
        \item $X$ is the set of all transactions that could be included in the blockchain as a result of the application protocol.
        \item $\Sigma_\ag$ is the set of \emph{strategy profiles}, where a strategy profile $\sigma_\ag\in\Sigma_\ag$ is a vector of \emph{strategies} $(\sigma_{\ag,i})_{i\in N}$, where a strategy $\sigma_{\ag,i}$ specifies what a player $i\in N$ does at each point in the game.
        \item $\pi_\ag:\Sigma_\ag\to \Y(X)$ is the \emph{outcome function}, which maps each $\sigma_\ag\in\Sigma_\ag$ to a set of transaction triples. 
    \end{itemize}
\end{definition}

\begin{definition}
    A \emph{parametrised application game} is a collection $\G=(\G^p)_{p\in\P}$, where $\P$ is some parameter set, and $\G^p$ is an application game $(N,X^p,\Sigma_\ag^p,\pi_\ag^p)$ for each $p\in\P$.
\end{definition}
\begin{remark}
\label{rem:dependence}
    For a given set $X$ of transactions, we can interpret the collection $(\bgame(Y))_{Y\in\Y(X)}$ as a parametrised blockchain game with parameter space $\Y(X)$. We may then denote $\bgame(Y)=(M,Y,\Sigma_{\bg}^Y,\pi_{\bg}^Y)$ to clearly indicate the dependency of $\Sigma_{\bg}$ and $\pi_{\bg}$ on~$Y$. 
\end{remark}

Remark~\ref{rem:dependence} shows how one can use parametrised games. Indeed, a protocol (or a blockchain, in this case) cannot be modelled by a single game. Depending on the context, such as the incoming transactions, or in the case of a protocol, e.g., when players learn certain information, the possible actions change. Making a game's parameter a function of what happens in another game will be what enables us to capture the interplay between different layers (Section~\ref{subsec:framework}) and different applications (Section~\ref{sec:composition}). Therefore, from now on, we only consider protocols $\Pi$ that can be written as a tuple $(N,(\G^p)_{p\in\P},(\overline{\sigma}_\ag^p)_{p\in\P})$.\FullOnly{ We also assume that there will always at most one player in $N$ who can broadcast a specific transaction and alter its fee. To this end, we define for each $p\in\P$ a so-called \emph{player function}.
\begin{definition}[Player function]
\label{def:player-function}
    Given a parametrised application game $(\G^p)_{p\in\P}$ with $\G^p=(N,X^p,\Sigma_\ag^p,\pi_\ag^p)$, we define for each $p\in\P$ the \emph{player function} $\chi^p:N\to\Po(X^p)$, which maps each player $i$ to the set of transactions for which $i$ pays the fees, such that $\chi^p(i)\cap\chi^p(j)=\varnothing$ for all $i,j\in N$.  
\end{definition}
}

\FullOnly{
    \subsection{The Network Game}
    \label{subsec:network}
    The application game outputs a set of transaction triples, but these do not immediately become inputs to the blockchain game. In practice, transactions are relayed over a communication network, whose structure determines how and when they are observed. For instance, in a synchronous setting with a global clock, all blockchain participants see the same transactions at the same time. In contrast, in more realistic settings, where mempool visibility may vary and adversarial actors can selectively delay or withhold propagation, participants may receive different subsets of transactions at different times, leading to diverging views and strategic behaviour.
    
    The network model we choose induces a so-called \emph{network game}. Players in this game could be protocol participants, as well as players in the blockchain game. For example, the L2 participants might have the option to share their transactions only with specific blockchain game players\footnote{Think for example, in the context of Bitcoin, about services that allow for the sharing of a transaction with one specific miner, such as the Mempool Accelerator (\url{https://mempool.space/accelerator}) or Marathon Slipstream (\url{https://slipstream.mara.com/})}. This, together with network assumptions, may lead to players in the blockchain game having different perspectives on the outcome of the application game. Moreover, in settings where the ordering of transactions within blocks in the blockchain might be relevant, such as Maximal Extractable Value (MEV)~\cite{daian2020flash}, a flexible network game is crucial.
    \begin{definition}[Network Game]
    \label{def:n-game}
        Let $Y\in\Y(X)$ be a set of transaction triples for some set $X$ of transactions. 
        The \emph{network game $\No(Y)$ induced by $Y$} is a tuple $(N,M,Y,\Sigma_{\nog},\pi_{\nog})$, where
        \begin{itemize}[leftmargin=*]
            \item $N$ is the set of $n=|N|$ application protocol participants, referred to as \emph{(protocol) players}.
            \item $M$ is the set of miners/validators, which are the \emph{players} of the blockchain game.
            \item $\Sigma_{\nog}$ is the set of \emph{strategy profiles}, where a strategy profile $\sigma_{\nog}\in\Sigma_{\nog}$ is a vector of \emph{strategies} $(\sigma_{\nog,i})_{i\in N\cup M}$, where a strategy $\sigma_{\nog,i}$ specifies what a player $i\in N\cup M$ does at each point in the game.
            \item $\pi_{\nog}=(\pi_{\nog,j})_{j\in M}$ is the \emph{outcome function}, where $\pi_{\nog,j}:\Sigma_{\nog}\to \Y(X)$ maps each $\sigma_{\nog}\in\Sigma_{\nog}$ to the set of transaction triples shared with $j$ for the blockchain game.
        \end{itemize}
    \end{definition}    
    \begin{remark}
    \label{rem:n-game}
    In general, blockchain players may receive different sets of transaction triples. Definition~\ref{def:b-game} extends naturally by taking as input a tuple $(Y_j)_{j \in M}$, with each player $j$ receiving only $Y_j$. The resulting blockchain game $\bmodel_M((Y_j)_{j \in M})$ produces orderings based on individual views. 
    \end{remark}
}

\section{Composition}
\label{sec:total-composition}
In the previous section we have introduced game-theoretic representations of the different layers of a blockchain ecosystem. We will now discuss how these representations relate to each other, and how we can form a game in the traditional sense out of them.

\subsection{Cross-layer Composition}
\label{subsec:framework}
Protocols deployed on a blockchain rarely operate in isolation. Instead, their security depends not only on the protocol logic at the application level but also on how these applications interact with the underlying network and consensus layers. For instance, a protocol that assumes timely inclusion of transactions may become insecure if miners have incentives to censor or reorder them. To reason about such interactions systematically, we introduce the \emph{cross-layer game}—a formal construction that \emph{sequentially} composes layers into a unified model.

This composition allows us to evaluate how strategies in one layer, such as transaction selection by miners, affect outcomes and incentives in other layers. On a high level, the cross-layer game's structure can be visualised by Figure~\ref{fig:layers}. That is, first the players of the application interact and a set of transactions triples is output, which then gets relayed according to the behaviour in the network\FullOnly{, modelled by playing the network game}\CameraOnly{\footnote{Since we assumed a synchronous network model with instantaneous message delivery and a global clock in this work, the output of the application game serves directly as the input set of transaction triples to the blockchain game.}}, and lastly get ordered into blocks according to the blockchain game. 
%
%
%
At the core of this construction is a utility function that maps the blockchain ordering \(b\) coming out of the blockchain game, which is the result of the strategic behaviour throughout application and blockchain layers, to payoffs for all players in these layers. This translation from ordering to utility is controlled by an execution function $\omega$. The cross-layer game allows us to reason about the incentive compatibility of a protocol with respect to a block behaviour \(\beta\), which we assume can be instantiated as a blockchain game \(\bgame(Y)\) for any set of transaction triples \(Y\), and an execution function.

\CameraOnly{
    \begin{definition}[(Parametrised) Cross-layer Game]
    \label{def:fgame}
        Consider a protocol $\Pi=\qty(N,(\G^p)_{p\in\P},(\overline{\sigma}_\ag^p)_{p\in\P})$ as defined in Section~\ref{subsec:application-game} with parametrised application game $(\G^p)_{p\in\P}$, a block behaviour $\beta$ that can be modelled by a blockchain game $\B^\beta$ with player set $M$ (for ease of notation we write $\B$ instead of $\B^\beta$), and an execution function $\omega$. The \emph{parametrised cross-layer game with respect to $\beta$ and $\omega$ of the protocol $\Pi$} is a collection $(\F^p_{\Pi,\beta,\omega})_{p\in\P}$, where for each $p\in\P$, $\F^p_{\Pi,\beta,\omega}$ is a \emph{cross-layer game with respect to $\beta$ and $\omega$ of the protocol $\Pi$}, which is defined as a tuple $(\G^p,\B(\G^p),\omega)$, where
        \begin{itemize}[leftmargin=*]
            \item $\G^p=(N,X^p,\Sigma_\ag^p,\pi_\ag^p)$ is an application game, which yields for every strategy $\sigma_\ag\in\Sigma_\ag^p$ a set $\pi_\ag^p(\sigma_\ag)$ of transaction triples. 
            \item $\B(\G^p)=\qty(\bgame(\pi_\ag^p(\sigma_\ag)))_{\sigma_\ag\in\Sigma_\ag^p}$ is a collection of blockchain games specified by $\beta$, where $\bgame(\pi_\ag^p(\sigma_\ag))=\Big(M,\pi_\ag^p(\sigma_\ag),\allowbreak\Sigma_{\bg}(\sigma_{\ag}),\allowbreak\pi_{\bg}^{(\sigma_\ag)}\Big)$\footnote{We use the shorthand $\Sigma_{\bg}(\sigma_{\ag}):=\Sigma_{\bg}^{\pi_\ag^p(\sigma_\ag)}$ (similar for $\sigma_{\bg}$) and $\pi_{\bg}^{(\sigma_\ag)}:=\pi_{\bg}^{\pi_\ag^p(\sigma_\ag)}$.} for each $\sigma_\ag\in\Sigma_\ag^p$. This outputs for each $\sigma_\ag\in\Sigma_\ag^p$ and $\sigma_{\bg}\in\Sigma_{\bg}(\sigma_{\ag})$ an output probability measure $P_{\sigma_{\bg}}$ on $\bmodel_M(\pi^p_\ag(\sigma_\ag))$.
            \item $\omega=(\omega_i)_{i\in N\cup M}$ specifies the \emph{execution function} for each player. This defines moreover a utility function $u_i^p$ for each $i\in N\cup M$. Indeed, first observe that the tuple $\sigma_\cg=(\sigma_\ag,\sigma_{\bg}(\cdot))$\footnote{We use the shorthand $\sigma_{\bg}(\cdot):=(\sigma_{\bg}(\sigma_{\ag}'))_{\sigma_{\ag}'\in\Sigma_{\ag}^p}$.} can be seen, up to natural identification, as a strategy profile of the cross-layer game. We therefore denote the set of cross-layer strategy profiles by
            \begin{multline*}
                \Sigma_\cg^p:=\Big\{(\sigma_\ag,\sigma_{\bg}(\cdot))\,\Big|\,\sigma_\ag\in\Sigma_\ag^p,\forall\sigma_\ag'\in\Sigma_\ag^p:\\ \sigma_{\bg}(\sigma'_{\ag})\in\Sigma_{\bg}(\sigma'_{\ag})\Big\}.
            \end{multline*}
            The utility function for player $i\in N\cup M$ can then be defined as a function $u^p_i:\Sigma_\cg^p\to\R$ which maps\footnote{We use measure-theoretic notation for generality and conciseness. In particular, $P_{\sigma_{\bg}}$ will generally be a discrete measure.}
            \begin{equation}
            \label{eq:utility}
                (\sigma_\ag,\sigma_{\bg}(\cdot)) \mapsto \int_{\bmodel_M(\pi_\ag^p(\sigma_\ag))}\omega_i(b)\dd P_{\sigma_{\bg}(\sigma_{\ag})}(b).
            \end{equation}
        \end{itemize}
    \end{definition}

    We have now defined an actual game, where first, for any $p\in\P$, players in $N$ choose a strategy profile $\sigma_\ag \in \Sigma_\ag^p$ in the \emph{application game}, and players in $M$ determine a strategy profile $\sigma_{\bg}(\sigma_\ag') \in \Sigma_{\bg}(\sigma_\ag')$ for every $\sigma_\ag'\in\Sigma_\ag^p$. The \emph{utility function} $u$ then assigns a utility to each player in $N \cup M$, given the combined strategy profile $\sigma_\cg = (\sigma_\ag,(\sigma_{\bg}(\sigma_\ag'))_{\sigma_\ag'\in\Sigma_\ag^p})$. 
}

\FullOnly{
    \begin{definition}[(Parametrised) Cross-layer Game]
    \label{def:fgame}
        Consider a protocol $\Pi=\qty(N,(\G^p)_{p\in\P},(\overline{\sigma}_\ag^p)_{p\in\P})$ as defined in Section~\ref{subsec:application-game} with parametrised application game $(\G^p)_{p\in\P}$, a block behaviour $\beta$ that can be modelled by a blockchain game $\B^\beta$ with player set $M$ (for ease of notation we write $\B$ instead of $\B^\beta$), and an execution function $\omega$. The \emph{parametrised cross-layer game with respect to $\beta$ and $\omega$ of the protocol $\Pi$} is a collection $(\F^p_{\Pi,\beta,\omega})_{p\in\P}$, where for each $p\in\P$, $\F^p_{\Pi,\beta,\omega}$ is a \emph{cross-layer game with respect to $\beta$ and $\omega$ of the protocol $\Pi$}, which is defined as a tuple $(\G^p,\No(\G^p),\B(\No),\omega)$, where
        \begin{itemize}[leftmargin=*]
            \item $\G^p=(N,X^p,\Sigma_\ag^p,\pi_\ag^p)$ is an application game, which yields for every strategy $\sigma_\ag\in\Sigma_\ag^p$ a set $\pi_\ag^p(\sigma_\ag)$ of transaction triples. 
            \item $\No(\G^p)=\qty(\No(\pi_\ag^p(\sigma_\ag)))_{\sigma_\ag\in\Sigma_\ag^p}$ is a collection of network games, where $\No(\pi_\ag^p(\sigma_\ag))=\Big(N,M,\pi_\ag^p(\sigma_\ag),\Sigma_{\nog}(\sigma_{\ag}),\allowbreak\pi_{\nog}^{(\sigma_\ag)}\Big)$\footnote{We use the shorthands $\Sigma_{\nog}(\sigma_{\ag}):=\Sigma_{\nog}^{\pi_\ag^p(\sigma_\ag)}$ (similar for $\sigma_{\nog}$) and $\pi_{\nog}^{(\sigma_\ag)}:=\pi_{\nog}^{\pi_\ag^p(\sigma_\ag)}$.} for each $\sigma_\ag\in\Sigma_\ag^p$. This outputs for each $\sigma_\ag\in\Sigma_\ag^p$ and $\sigma_{\nog}\in\Sigma_{\nog}(\sigma_{\ag})$ a collection $(\pi_{\nog,j}^{(\sigma_{\ag})}(\sigma_{\nog}))_{j\in M}$, indicating which transaction triples are shared with each miner $j\in M$. We denote by $\Sigma^p_{\ag,\nog}$ the set of all tuples $(\sigma_{\ag},\sigma_{\nog}(\cdot))$, where $\sigma_{\ag}\in\Sigma_{\ag}^p$ and $\sigma_{\nog}(\cdot)$ is a shorthand for $(\sigma_{\nog}(\sigma_{\ag}'))_{\sigma_{\ag}'\in\Sigma_{\ag}^p}$ such that $\sigma_{\nog}(\sigma_{\ag}')\in\Sigma_{\nog}(\sigma_{\ag}')$. The latter essentially encapsulates the strategies taken in the network game for each possible $\sigma_{\ag}$ that may be taken in the application game.
            \item $\B(\No)=\qty(\bgame\big(\pi_\nog^{(\sigma_{\ag})}(\sigma_\nog(\sigma_{\ag}))\big))_{(\sigma_\ag,\sigma_{\nog}(\cdot))\in\Sigma_{\ag,\nog}^p}$ is the collection of blockchain games specified by $\beta$, where $\bgame\big(\pi_\nog^{(\sigma_{\ag})}(\sigma_\nog(\sigma_{\ag}))\big)=\Big(M,\pi_\nog^{(\sigma_{\ag})}(\sigma_\nog(\sigma_{\ag})),\allowbreak\Sigma_{\bg}(\sigma_{\ag},\sigma_{\nog}(\cdot)),\allowbreak\pi_{\bg}^{(\sigma_{\ag},\sigma_{\nog}(\cdot))}\Big)$ for each $(\sigma_\ag,\sigma_{\nog}(\cdot))\in\Sigma_{\ag,\nog}^p$. For one such blockchain game, each $\sigma_{\bg}\in\Sigma_{\bg}(\sigma_{\ag},\sigma_{\nog}(\cdot))$ implies an output probability measure $P_{\sigma_{\bg}}$ on $\bmodel_M(\pi_\nog^{(\sigma_{\ag})}(\sigma_\nog(\sigma_{\ag})))$.
            \item $\omega=(\omega_i)_{i\in N\cup M}$ specifies the \emph{execution function} for each of the players. This defines moreover a utility function $u_i^p$ for each $i\in N\cup M$. Indeed, first observe that the tuple $\sigma_\cg=(\sigma_\ag,\sigma_{\nog}(\cdot),\sigma_{\bg}(\cdot,\cdot))$\footnote{We use the shorthand $\sigma_{\bg}(\cdot,\cdot):=(\sigma_{\bg}(\sigma_{\nog}(\sigma_{\ag})))_{(\sigma_\ag,\sigma_{\nog}(\cdot))\in\Sigma_{\ag,\nog}^p}$.} can be seen, up to natural identification, as a strategy profile of the cross-layer game. We therefore denote the set of cross-layer strategy profiles
            \begin{footnotesize}
            \begin{equation*}
                \left\{
                (\sigma_\ag,\sigma_{\nog}(\cdot),\sigma_{\bg}(\cdot,\cdot))
                \;\middle|\;
                \begin{aligned}
                    & \sigma_\ag\in\Sigma_\ag^p, \\
                    & \forall\sigma_\ag'\in\Sigma_\ag^p:\sigma_{\nog}(\sigma'_{\ag})\in\Sigma_{\nog}(\sigma'_{\ag}), \\
                    & \forall\sigma_{\ag,\nog}'\in\Sigma_{\ag,\nog}^p:\sigma_{\bg}(\sigma_{\ag,\nog}')\in\Sigma_{\bg}(\sigma_{\ag,\nog}')
                \end{aligned}
                \right\}
            \end{equation*}
            \end{footnotesize}
            by $\Sigma_\cg^p$, where we used the shorthand $\sigma_{\ag,\nog}'=(\sigma_{\ag}',\sigma_{\nog}'(\cdot))$. The utility function for player $i\in N\cup M$ can then be defined as a function $u^p_i:\Sigma_\cg^p\to\R$ which maps\footnote{We use measure-theoretic notation for generality and conciseness. In particular, $P_{\sigma_{\bg}}$ will generally be a discrete measure.}
            \begin{footnotesize}
            \begin{equation}
            \label{eq:utility}
                (\sigma_\ag,\sigma_{\nog}(\cdot),\sigma_{\bg}(\cdot,\cdot)) \mapsto \int_{\bmodel_M(\pi_\nog^{(\sigma_{\ag})}(\sigma_\nog(\sigma_{\ag})))}\omega_i(b)\dd P_{\sigma_{\bg}(\sigma_{\nog}(\sigma_{\ag}))}(b).
            \end{equation}
            \end{footnotesize}
        \end{itemize}
    \end{definition}

    We have now defined an actual game, where first, for any $p\in\P$, players in $N$ choose a strategy profile $\sigma_\ag \in \Sigma_\ag^p$ in the \emph{application game}, both players in $N$ and $M$ choose a network game strategy profile $\sigma_{\nog}(\cdot)$, considering all possible $\sigma_\ag \in \Sigma_\ag^p$, and players in $M$ determine a blockchain game strategy profile $\sigma_{\bg}(\cdot,\cdot)$, considering all possible $\sigma_{\ag,\nog}\in\Sigma_{\ag,\nog}^p$. The \emph{utility function} $u$ then assigns a utility to each player in $N \cup M$, given the combined strategy profile $(\sigma_\ag,\sigma_{\nog}(\cdot),\sigma_{\bg}(\cdot,\cdot))$. 

    \begin{remark}
        Henceforth, we assume a synchronous network model with instantaneous message delivery and a global clock, as is common in prior work~\cite{garay2024bitcoin,kiayias2017ouroboros}. Hence, all blockchain players observe the same application game outcome. That is, we assume for any set $X$ of transactions and set $Y$ of transaction triples in $\Y(X)$ the network game with $\Sigma_{\nog}=\qty{\sigma_{\nog}}$ such that $\pi_{\nog,j}(\sigma_{\nog})=Y$ for each $j\in M$. This allows us to simplify notation by fully omitting the network game, except in Section~\ref{subsec:mev}, where MEV motivates a richer model. We can therefore apply all shorthands for network game components instead to blockchain game components, e.g., writing an element of $\Sigma_{\cg}^p$ now as $(\sigma_{\ag},\sigma_{\bg}(\cdot))$.
    \end{remark}
}

\begin{remark}
    \pim{please check}One limitation of Definition~\ref{def:fgame} is that the cross-layer composition is sequential. In our effort to analyse the application and blockchain layers separately, certain scenarios cannot be captured by the framework. In particular, application game players cannot strategise based on what is the actual state of the blockchain. For example, if at a certain point in the application game, a player Alice posts two conflicting transaction triples $(\tx{1}{},t,f)$ and $(\tx{2}{},t,f)$, another player Bob has to make a decision based on whether he thinks $\tx{1}{}$ or $\tx{2}{}$ will be ordered first. However, it would be more accurate for Bob to strategise based on the transaction that \emph{actually} gets ordered first. It is left for future work to determine whether such scenarios could be artificially captured by the current framework, and, if not, whether it is at all possible to model the above scenarios while keeping the analysis of application and blockchain games separate.
\end{remark}

To analyse the \emph{game-theoretic security} of a protocol $\Pi=\qty(N,(\G^p)_{p\in\P},(\overline{\sigma}_\ag^p)_{p\in\P})$ with respect to a given $\beta$, we must construct the game $\F^p_{\Pi,\beta,\omega}$ for each $p\in\P$ and examine its utility function. In particular, we are interested in how the \emph{utility changes when players deviate from the IPB}.
\begin{definition}[Deviation]
\label{def:deviation}
    Consider a cross-layer game $\F^p_{\Pi,\beta,\omega}$ with strategy profile space $\Sigma_\cg^p$. Let $\sigma_\cg,\sigma_\cg'\in\Sigma_\cg^p$ be two arbitrary strategy profiles. A strategy profile $\sigma_\cg$ can be interpreted as a vector of strategies for all players in the cross-layer game, i.e., $\sigma_\cg=\qty(\sigma_{\cg,i})_{i\in N\cup M}$, where $\sigma_{\cg,i}=\sigma_{\ag,i}$ for $i\in N\setminus M$, $\sigma_{\cg,i}=\sigma_{\bg,i}(\cdot)$ for $i\in M\setminus N$, and $\sigma_{\cg,i}=(\sigma_{\ag,i},\sigma_{\bg,i}(\cdot))$ for $i\in N\cap M$. A strategy $\sigma_\cg'= (\sigma'_{\cg,i})_{i\in N\cup M}$ is called a \emph{deviation from $\sigma_\cg$ by $i$} if $\sigma_{\cg,i}\neq\sigma'_{\cg,i}$ and $\sigma_{\cg,j}=\sigma'_{\cg,j}$ for all $j\in (N\cup M)\setminus\qty{i}$.
\end{definition}
A key goal of our framework is to reason about collusion without sacrificing composability. To this end, we formalise a reduction mechanism that preserves strategic structure and utility. While our initial notion of deviation focuses on unilateral actions, real-world settings often involve coordinated deviations by coalitions. We define a coalition (or collusion) as any group of players acting as one, possibly including multiple application-layer players and one blockchain-layer player. The latter simplifies exposition and is without loss of generality for many blockchain games. For instance, in U-COBG, colluding miners $j,j' \in M$ with hashrates $\lambda_j, \lambda_{j'}$ can be treated as a single miner with hashrate $\lambda_j + \lambda_{j'}$ in a modified game with $\lambdavec=(\lambda_0,\ldots,\lambda_j,\ldots,\lambda_{j'},\ldots,\lambda_m)$.

To model collusions formally, we introduce an \emph{$\eta$-collusion reduction (CR)}, which maps a game $\Ga$ to a reduced game $\tilde{\Ga}$ in which a set of colluding players is represented by a single aggregated player. The function $\eta$ specifies how coalitions in $\Ga$ are mapped to unified players in $\tilde{\Ga}$, enabling structured analysis of coalition strategies while preserving the utility semantics of the original game. Note that the game $\Ga$ can refer to an \CameraOnly{application/blockchain}\FullOnly{application/network/blockchain} game, or most likely, a cross-layer composition of different layer games. Due to our construction of the layer games, even a cross-layer composition will always have the form $(N,*,\Sigma,\pi)$, where $N$ is some player set, $\Sigma$ a strategy profile set, and $\pi$ some kind of outcome function. $*$ indicates any additional layer-specific components that would be present.


\begin{definition}[$\eta$-CR]
\label{def:eta-collusion-reduction}   
    Let $\Ga=\qty(N,*,\Sigma,\pi)$ and $\tilde{\Ga}=(\tilde{N},*,\tilde{\Sigma},\tilde{\pi})$ be two application games. Let $\eta:N\rightarrow N$ be a transformation of $N$. We say that $\tilde{\Ga}$ is an \emph{$\eta$-collusion reduction (CR) of $\Ga$} if $\tilde{N}=N$, $\tilde{\Sigma}=\Sigma$, $\tilde{\pi}=\pi$ and any other additional components are also equal. Every player $i\in\eta(N)$ now decides on a strategy $\sigma_{i}=(\sigma_{j})_{j\in\eta^{-1}(i)}$, i.e., on the strategies of all $j\in N$ for which $\eta(j)=i$. 
\end{definition}
\begin{remark}
\label{rem:eta-one-miner}
    We assume from now on that if the player set $N$ contains a subset $N_{\bg}$ of players from a blockchain game, every $\eta:N\to N$ satisfies $\eta(j)=j$ for all $j
    \in N_{\bg}$. 
\end{remark}
Note that all the players in $N$ are still part of the $\eta$-CR $\tilde{\Ga}$, except some of them will have no actions to take. The $\eta$-CR $\tilde{\Ga}$ of a game $\Ga$ can again be used to construct---or already is---a cross-layer game. Recalling \eqref{eq:utility}, the utility of any player in $i\in\tilde{N}$ is the sum of the utilities of the players in $\eta^{-1}(i)$. This effectively defines the CR of a cross-layer game, from which we can define the CR of a protocol.
\begin{definition}[CR of a protocol]
\label{def:collusion-reduced}
    Consider two instances $\Pi=(N,(\G^p)_{p\in\P},(\overline{\sigma}_\ag^p)_{p\in\P})$ and $\tilde{\Pi}=(N,(\tilde{\G}^p)_{p\in\P},\allowbreak(\overline{\sigma}_\ag^p)_{p\in\P})$ of the same protocol, leading with a block behaviour $\beta$ (implying a player set $M$) and an execution function $\omega$ to respective collections of cross-layer games $(\G^p,\B(\G^p),\omega)_{p\in \P}$ and $(\tilde{\G}^p,\B(\tilde{\G}^p),\omega)_{p\in\P}$. We say that $\tilde{\Pi}$ is a \emph{CR of $\Pi$}, if there exists an $\eta:N\cup M\rightarrow N\cup M$ such that for every $p\in\P$, the cross-layer game $(\tilde{\G}^p,\B(\tilde{\G}^p),\omega)$ is a CR, more specifically an $\eta$-CR, of the cross-layer game $(\G^p,\B(\G^p),\omega)$.
\end{definition}
The CR of a protocol groups each collusion of players into a single player, responsible for the decisions of all members. Definition~\ref{def:deviation} naturally applies to the CR of a cross-layer game: a deviation in the CR corresponds to multiple deviations in that cross-layer game.

Since blockchain game players are assumed not to collude, we can treat the blockchain game as a black box. Given a set of transaction triples, it outputs a probability measure over blockchain orderings. Intuitively, we can pre-compute the strategy profiles miners or validators would adopt for any set of transaction triples. 
\begin{definition}[Optimal strategy sets of $X$ w.r.t. $\beta$]
\label{def:xi-star}
     For a given blockchain $\beta$, an execution function $\omega$ and a set $Y\in\Y(X)$, we define the \emph{completed blockchain game} $(\bgame(Y),\omega)$. This is just the blockchain game induced by $Y$, together with the utility function $u=(u_j)_{j\in M}$, where $u_j:\Sigma_{\bg}\rightarrow\R_{\geq0}$ for $j\in M$ is defined, similarly to \eqref{eq:utility}, as the expected fees obtained by player $j$, together with any other funds that can be claimed by $j$ according to the execution function, i.e., $u_j(\sigma_{\bg})=\int_{\bmodel_M(Y)}\omega_j(b)\dd P_{\sigma_{\bg}}(b)$. One can compute the set of strategy profiles $\overline{\Sigma}_{\bg}^{Y}$, 
     containing all (mixed-strategy) Nash equilibria of $(\bgame(Y),\omega)$, after iterated removal of weakly dominated strategies\footnote{We assume that the blockchain model leads to a completed blockchain game such that this set is non-empty.}. We call $\overline{\Sigma}_{\bg}^{Y}$ the \emph{optimal strategy set for $Y$ w.r.t. $\beta$}. Similarly, we call elements of $\overline{\Sigma}_{\bg}^{Y}$ the \emph{optimal strategy profiles for $Y$ w.r.t. $\beta$} and denote them by $\overline{\sigma}_{\bg}$\footnote{Note that $\overline{\sigma}_{\bg}$ could be a mixed strategy, but it will still yield a probability distribution on the appropriate space of blockchain orderings and is thus well-defined within our framework.}. 
     The collection $(\overline{\Sigma}_{\bg}^{Y})_{Y\in\Y(X)}$ is the \emph{collection of optimal strategy sets of $X$}. Keep in mind that an (optimal) strategy profile in the cross-layer game would be a collection of (optimal) strategy profiles for $Y$, for all $Y$ that could be the outcome of the application game. 
\end{definition}
\begin{remark}
\label{rem:one-miner-optimum}
For ease of presentation, we will henceforth assume that for a given block behaviour $\beta$ (implying player set $M$), an execution function $\omega$ and a set $Y\in\Y(X)$, $|\overline{\Sigma}_{\bg}^Y|=1$, i.e., there will be a unique optimal strategy profile, denoted by $\overline{\sigma}_{\bg}^Y$ for $Y$ w.r.t. $\beta$. This allows us to define a utility function immediately for an application game $\G=(N,X,\Sigma_\ag,\pi_\ag)$ as follows. For $i\in N\cup M$, let $u_i:\Sigma_\ag\to\R$ map $\sigma_\ag\mapsto\int_{\bmodel_M(\pi_\ag(\sigma_\ag))}\omega_i(b)\dd P_{\overline{\sigma}_{\bg}(\sigma_{\ag})}(b)$. We denote the \emph{completed application game w.r.t. $\beta$ and $\omega$} as $(\G,\beta,\omega)$. The definition of a CR easily extends to this completed application game. This assumption is strong, as existence and uniqueness (especially the latter) of game-theoretic equilibrium concepts are not guaranteed, especially in games that become more complex in trying to better approximate reality. Future research on more complex models for block behaviour should be accompanied by investigating the existence, and possible uniqueness of optimal strategy profiles.  
\end{remark}

Note that the notions of completed blockchain and application games allow us to consider individual layers as stand-alone games that define a utility, where the other layers have been abstracted. For completed blockchain games, we start our analysis with a set of transaction triples, not having to worry about how this set was obtained through strategic interactions in the application layer. In the case of completed application games, the outcome is directly mapped to a utility via a black box abstracting away a blockchain game with optimal play. Since the optimal play is assumed to exist and to be unique, this mapping is well-defined. 


As discussed earlier, a protocol is considered game-theoretically secure if rational players, or coalitions thereof, have no incentive to deviate from the prescribed strategy profile. For a given set of transactions $X$ and a block behaviour $\beta$, we have defined the optimal strategy sets of $X$ with respect to $\beta$, allowing us to describe how rational blockchain players will respond to any set of transaction triples in $\Y(X)$. With these components in place, we can define the notion of \emph{incentive compatibility} for a protocol.

\begin{definition}[IC w.r.t. $\beta$]
\label{def:incentive-compatibility}
    Let $\Pi=(N,(\G^p)_{p\in\P},\allowbreak(\overline{\sigma}_\ag^p)_{p\in\P})$ be a protocol, $\beta$ a block behaviour (implying player set $M$), $\omega$ an execution function and $(\F^p_{\Pi,\beta,\omega})_{p\in\P}$ the implied collection of cross-layer games. We say that $\Pi$ is \emph{IC w.r.t. $\beta$} if for each $p\in\P$, the pair $((\G^p,\beta,\omega),\overline{\sigma}_\ag^p)$ is \emph{IC w.r.t. $\beta$}. That is, after removal of weakly dominated strategies, the strategy profile $\overline{\sigma}_\ag^p$ is a \emph{Nash equilibrium} for every CR of $(\G^p,\beta,\omega)$. In other words, for each $\eta:N\cup M\rightarrow N\cup M$, we have for each $i\in\eta(N\cup M)$ and deviation $\sigma_\ag$ from $\overline{\sigma}_\ag^p$ by $i$, that $u_i(\overline{\sigma}_\ag^p)\geq u_i(\sigma_\ag)$.    

    
\end{definition}

\subsection{Cross-application Composition}
\label{sec:composition}
A central goal of our framework is to reason about the incentive security of protocols not just in isolation, but when deployed together. We refer to this property as \emph{security under composition}. 
While the previous section focused on a single protocol’s interaction with the blockchain, real-world systems often involve multiple applications operating concurrently. These may interfere with each other, affecting both player strategies and outcome.
To capture such settings, we now formalise the composition of parametrised application games.
For simplicity, we assume that the sets of transactions \(X_1\) and \(X_2\) generated by two protocols \(\Pi_1\) and \(\Pi_2\) are disjoint; otherwise, we would no longer be modelling two distinct protocols, but rather a single unified one.

\begin{definition}[$\gvec$-composition of parametrised application games]
\label{def:composition}
    Given the parametrised application games $\G_1:=(\G_1^p)_{p\in\P}$ and $\G_2:=(\G_2^q)_{q\in\Q}$, where for any $p\in\P$, $\G_1^p=(N_1,X_1^p,\Sigma_{\ag,1}^p,\pi_{\ag,1}^p)$ and any $q\in\Q$, $\G_2^q=(N_2,X_2^q,\Sigma_{\ag,2}^q,\pi_{\ag,2}^q)$, and $\gvec$ is a collection of functions $\gvec=(g_p)_{p\in\P}$, where $g_p:\Sigma_{\ag,1}^p\to\Q$, we define the \emph{$g$-composition of $\G_2$ and $\G_1$}, denoted by $\G_2\circ_\gvec\G_1$, as the parametrised application game $\G_2\circ_\gvec\G_1=(\G^p)_{p\in\P}$, where for any $p\in\P$, $\G^p=\qty(N,X^p,\Sigma_\ag^p,\pi_\ag^p)$ is defined by 
    \begin{itemize}[leftmargin=*]
        \item $N=N_1\cup N_2$ is the set of players,
        \item $X^p=X^p_1\cup\bigcup_{q\in g_p(\Sigma_{\ag,1}^p)}X^q_2$ is the set of transactions,
        \item $\Sigma_\ag^p$ is the set of strategy profiles, which up to natural identification equals $\Sigma_{\ag,1}\times\prod_{\sigma_{\ag,1}\in\Sigma_{\ag,1}}\Sigma_{\ag,2}^{g_p(\sigma_{\ag,1})}$. 
        \item $\pi_\ag^p:\Sigma_\ag^p\to\Y(X^p)$ is the outcome function, mapping each $\sigma_\ag=\allowbreak(\sigma_{\ag,1},(\sigma_{\ag,2,\sigma_\ag})_{\sigma_\ag\in\Sigma_{\ag,1}})\in\Sigma_\ag^p$ to $\pi_{\ag,1}^p(\sigma_{\ag,1})\cup\pi_{\ag,2}^{g_p(\sigma_{\ag,1})}(\sigma_{\ag,2,\sigma_{\ag,1}})$.
    \end{itemize}    
\end{definition}

Intuitively, for a given $p\in\P$, the function $g_p$ tells us how the strategic decisions taken in the application game $\G_1^p$ influence the state of the application $\Pi_2$, i.e., which game from $\G_2$ actually gets played. $\gvec$-compositions enable us to model the interaction of protocols via means other than the blockchain. We illustrate this in Section~\ref{subsec:pc}, where a suitable $\gvec$ can model out-of-band communication of information required to take certain actions in a protocol.

Since the composition of two parametrised application games is again a parametrised application game, we can use it to construct a cross-layer game with respect to a block behaviour $\beta$. We can now say whether two protocols are \emph{IC w.r.t $\beta$ under $\gvec$-composition}. 
\begin{definition}[IC w.r.t. $\beta$ under $\gvec$-composition]
\label{def:secure-composition}
    Consider a block behaviour $\beta$, an execution function $\omega$, and two protocols $\Pi_1=\qty(N_1,(\G^p_1)_{p\in\P},(\overline{\sigma}^p_{\ag,1})_{p\in \P}), \Pi_2=\qty(N_2,(\G^q_2)_{q\in\Q},(\overline{\sigma}^q_{\ag,2})_{q\in\Q})$. 
    We say that $\Pi_1$ and $\Pi_2$ are \emph{IC w.r.t. $\beta$ under $\gvec$-composition} if $\Pi_1$ is IC w.r.t. $\beta$, $\Pi_2$ is IC w.r.t. $\beta$, and if for each $p\in\P$, $((\G^p,\beta,\omega),\overline{\sigma}_\ag^p)$ is IC w.r.t. $\beta$, where $(\G^p)_{p\in\P}=\G_2\circ_\gvec\G_1$ and for each $p\in\P$, $\overline{\sigma}_\ag^p$ is given, up to natural identification, by $\overline{\sigma}_\ag^p=\qty(\overline{\sigma}_{\ag,1}^p,(\overline{\sigma}_{\ag,2}^{g_p(\sigma_\ag)})_{\sigma_\ag\in\Sigma_{\ag,1}^p})$.
\end{definition}
Determining which protocols remain IC under \(\gvec\)-composition, and under what conditions on the blockchain and execution models, is a challenging problem. In general, such results require specific structural assumptions about the protocols, player utilities, and the environment. We now highlight preliminary results in the setting of \emph{additive} execution functions.  
\begin{definition}
\label{def:additive-execution}
    Consider a block behaviour $\beta$ (implying player set $M$) and an execution function $\omega$. We say that the execution function $\omega$ is \emph{additive}, if for any two different parametrised application games $(\G_1^p)_{p\in\P}$ and $(\G_2^q)_{q\in\Q}$, for every collection of functions $\gvec=(g_p)_{p\in\P}$, and for every $p\in\P$ and $i\in N\cup M$, the utility function $u^p_i:\Sigma_\ag^p\rightarrow\R$ of the completed application game $((\G_2\circ_\gvec\G_1)^p,\beta,\omega)$ is given by the map $(\sigma_{\ag,1},(\sigma_{\ag,2,\sigma_\ag'})_{\sigma_\ag'\in\Sigma_{\ag,1}})\mapsto u^p_{1,i}(\sigma_{\ag,1})\1_\qty{i\in N_1\cup M}+u^{g_p(\sigma_{\ag,1})}_{2,i}(\sigma_{\ag,2,\sigma_{\ag,1}})\1_\qty{i\in N_2\cup M}$, where $u^p_{1/2,i}$ is the utility function of the completed application game $(\G_{1/2}^p,\beta,\omega)$.
\end{definition}
Intuitively, we would expect realistic execution functions to be additive, assuming that once the input to the blockchain game is known, the validators will consider how to handle the transactions from different applications separately. Indeed, transactions from different applications should not be ``conflicting'' according to a reasonable execution function, as otherwise the two applications would in fact form one large application. It remains however a direction for future work to concretely formulate and categorise realistic execution functions, and determine whether they are additive. 

A simple result holds for compositions with collections $\gvec=(g_p)_{p\in\P}$, where for each $p\in\P$, $g_p$ is a constant function of $\sigma_\ag\in\Sigma_\ag^p$. That is, for a fixed $p\in\P$, regardless of the strategic decisions in the first game, the exact same version of the second game is played. That is, the first game has no effect on the second one. The result follows by contradiction: if there would be a profitable deviation from the IPB of the $\gvec$-composition, we could find a profitable deviation from the IPB of one of the original games, which is at odds with these games being IC.
\begin{theorem}[\ref{subsec:thm:constant-composition}]
\label{thm:constant-composition}
    With $\beta$ a block behaviour and $\omega$ an additive execution function, let $\Pi_1=(N_1,(\G^p_1)_{p\in\P},\allowbreak(\overline{\sigma}^p_{\ag,1})_{p\in \P})$ and $\Pi_2=(N_2,(\G^p_2)_{q\in\Q},(\overline{\sigma}^q_{\ag,2})_{q\in\Q})$ be two IC protocols w.r.t. $\beta$. Then $\Pi_1$ and $\Pi_2$ are IC w.r.t. $\beta$ under $\gvec$-composition, for each collection $\gvec=(g_p)_{p\in\P}$ of constant functions, i.e., where for each $p\in\P$, there exists some $q_p\in\Q$ such that $g_p(\sigma_{\ag,1})=q_p$ for each $\sigma_{\ag,1}\in\Sigma_{\ag,1}^p$.
\end{theorem}

\section{Case Studies} 
\label{sec:examples}
We present a series of case studies illustrating the applicability of our framework. We first define the class of \emph{timelock-bribe sensitive applications} and analyse their susceptibility to timelock bribing attacks as a function of blockchain predictability. This case study shows how the framework yields insights that apply to abstract classes of blockchain protocols. We then demonstrate how the framework captures common cross-layer and cross-application interactions, using several payment channel (PC)-based constructions
. Finally, we outline a broader class of applications, including cross-chain protocols and multi-layer mechanisms such as proposer-builder separation (PBS). Although not formalised here, these examples highlight the framework’s broader applicability to modelling incentive dynamics across heterogeneous blockchain protocols and deployment architectures.

\subsection{Timelock-Bribe Sensitive Applications}
\label{subsec:tbs}
Our framework avoids ad-hoc analyses for specific blockchain applications by enabling game-theoretic security statements for entire classes of applications. We demonstrate this by deriving necessary conditions for the incentive compatibility of \emph{timelock-bribe sensitive (TBS) applications}, using a Censor-Only Blockchain Game (COBG). Informally, TBS applications depend on the inclusion of specific transactions to function correctly, but allow deviations by posting conflicting, timelocked transactions, trying to censor the original transactions. We assume a Bitcoin-like execution function $\omega$, which enforces orderings aligned with Bitcoin’s rules (e.g., no spending the same coin twice, adhering to timelocks and signature requirements). Thus, we assume miners avoid proposing blocks with invalid transactions, as these yield suboptimal utility. 
\begin{remark}
    We \pim{please check}informally state that transactions are included in (or censored from) the blockchain ordering. This seems at odds with Definition~\ref{def:blockchain}, but is purely meant to ease the presentation. Blockchains like Bitcoin do not explicitly create an ordering of all input transactions. The Bitcoin ledger only contains a subset of these transactions, implicitly ordering all transactions not included on-chain in such a way that the execution function ignores them. Henceforth, we may state that a transaction $\tx{1}{}$ is included instead of $\tx{2}{}$. This formally means that both transactions are included in the transaction ordering, but that $\tx{1}{}$ is prioritised and $\tx{2}{}$ is ignored by the execution function. This interpretation of what is included in the blockchain does not change Bitcoin's properties. In Definition~\ref{def:timelock-based}, we explicitly state the rules the execution function follows, interpreting the blockchain ordering not as defined in the more general Definition~\ref{def:blockchain}, but instead as the sequence containing only the transactions actually included on-chain.
\end{remark}
\begin{definition}[Timelock-Bribe Sensitive Application]
\label{def:timelock-based}
    An application $\Pi=(N,(\G^p)_{p\in\P},(\overline{\sigma}_\ag^p)_{p\in\P})$ is called a \emph{timelock-bribe sensitive (TBS) application with timelock $T$ under $\omega$} if for some $p\in\P$, there exists a deviation $\sigma_\ag^p$ from $\overline{\sigma}_\ag^p$ such that $(\tx{1}{},t,f_1)\in\pi_\ag^p(\overline{\sigma}_\ag^p)$, $(\tx{2}{},t,f_2)\notin\pi_\ag^p(\overline{\sigma}_\ag^p)$ and $\qty{(\tx{1}{},t,f_1),(\tx{2}{},t,f_2)}\subseteq\pi_\ag^p(\sigma_\ag^p)$ where $\tx{1}{},\tx{2}{}\in X^p$ and where, if $\qty{\tx{1}{},\tx{2}{}}\subseteq b$ or if $b$ includes $\tx{2}{}$ before block height $T$, $\omega_i(b)=-\infty$ for any blockchain participant $i$. 
\end{definition}

The transaction triples output by a TBS application induce a probability distribution over possible blockchain orderings. To assess incentive compatibility, we must study this distribution: specifically, the probability of the event $C$ that $\tx{1}{}$ is censored and $\tx{2}{}$ is included instead.
Normally, when a miner would only see $\tx{1}{}$, it is rational to include it immediately. Since all miners would do so, $\tx{1}{}$ is included with probability one in the round $t_1$. In the presence of a valid conflicting $\tx{2}{}$, miners would include the transaction with the highest fee (invalidating the other). TBS applications, however, involve two mutually exclusive transactions with \emph{different timing constraints}: one immediately includable, the other timelocked. This tension has been studied in various mining models~\cite{avarikioti2022suborn, wang2022arbitrage, aumayr2024securing}. We now examine it under U- and G-COBG, deferring L-COBG to future work. Consider the minimal setting $Y = {(\tx{1}{}, t_1, f_1), (\tx{2}{}, t_2, f_2)}$, where $\tx{1}{}$ is available from time 0 and $\tx{2}{}$ from time $T$. Let $f_1 \in \R_{\geq 0}$ and $f_2$ be a fee function distributing a total \emph{bribe budget} $F_2$ among miners. We focus without loss of generality on the case that $t_1=t_2=0$ and $T>0$; in other cases either one transaction is visible before the other and will simply be included immediately, or both transactions are valid at the same time and the highest paying one is included. We also assume $\kappa>1$, where we write $\kappa=F_2/f_1$. If $\kappa\leq1$, miners would always prefer $\tx{1}{}$.
\smallskip\noindent\textbf{G-COBG.} The censoring probability $P(C)$ depends on how many distinct miners can be bribed slightly above $f_1$ to censor $\tx{1}{}$; if the bribe budget suffices to bribe every miner until round $T$, the rationality of censoring becomes clear once the leader schedule is known.
\begin{theorem}[\ref{subsec:thm:g-cobg}]
\label{thm:g-cobg}
    Consider the G-COBG $(\lambdavec,Y,D,\Sigma_{\Gup,\bg},\pi_{\Gup,\bg})$ with stake distribution $\lambdavec=(\lambda_j)_{j=0}^m$ and $Y=\qty{(\tx{1}{},0,f_1),(\tx{2}{},0,f_2)}$, such that $\tx{2}{}$ can be included at or after round $T>0$, where $T<D$, and where a bribe budget $F_2:=\kappa f_1$ is distributed by $f_2$ ($\kappa>1$), which has been fixed before the leader schedule for epoch $0$ is known. Then the probability that $\tx{1}{}$ is censored equals
    \begin{equation*}
        P(C)=(1-\lambda_0)^{T}\big((1-\lambda_0)p_{T}(\kappa)+\lambda_0p_{T-1}(\kappa-1)\big),
    \end{equation*}
    where we write $p_T(\kappa):=P(N_0^T<\kappa)$, with $N_0^T$ the number of distinct validators in $(v_t^0)_{t=0}^{T}$ given that there is no $t$ such that $v_t^0=0$, i.e., $N_{0}^{T}=\sum_{j=1}^m\1_\qty{\exists t\in\{0,\ldots,T\}:v_t=j}$.
\end{theorem}
\begin{remark}
\label{rem:N0T}
    The distribution of the number $N_0^T$ of distinct validators drawn with replacement from a set of $m$ validators with drawing probabilities $\lambda_1/(1-\lambda_0),\ldots,\lambda_m/(1-\lambda_0)$ does in general not have a closed-form expression, we can instead compute it numerically\footnote{\label{fot:link}\url{https://anonymous.4open.science/r/COBG}.}.
\end{remark}
As an example, in Figure~\ref{fig:p}, we show with the solid line how $P(C)$ evolves in G-COBG as a function of $\kappa$, for a toy hashrate distribution $\lambdavec'=\frac{1}{10}(2,1,3,4)$ and $T=8$. Note how the probability increases with $\kappa$ and never grows larger than $(1-\lambda_0)^{T}\approx0.17$.

\smallskip\noindent\textbf{U-COBG.} Analysing U-COBG turns out to be more complex. This is because at a given round $t$, miner $j=1,\ldots,m$ decides whether or not to censor $\tx{1}{}$ based on the expected fees it would theoretically gain from the blockchain history up to round $t$, and based on the fees it still expects to gain until round $T$, which all depends on which other miners are still censoring. More formally, we can formulate a recursive relation for the expected gain of miner $j=1,\ldots,m$ at round $t=0,\ldots,T$ under optimal play, assuming that $\tx{1}{}$ has not been included so far and that $h$ is the sequence of leaders in rounds $0,\ldots,t-1$:
\begin{equation}
\label{eq:u-cobg}
    g_j^t(h)=\lambda_j\qty(f_1\vee g_j^{t+1}(h\|j))+\sum_{k\in \Gamma^t\setminus\qty{j}}\lambda_k g_j^{t+1}(h\|k),
\end{equation}
where $\Gamma^t$ is the set of censoring miners in round $t$. The reasoning behind \eqref{eq:u-cobg} is that at round $t$, miner $j$ will choose to censor if the expected gain $g_j^{t+1}(h\|j)$ from being the leader in round $t$ and censoring is larger than the gain $f_1$ from including $\tx{1}{}$ instead. To compute the actual expected gain, one also needs to account for other censoring miners being the leader in round $t$ and thus computing the expected gain from round $t+1$, given the past leader sequence. From \eqref{eq:u-cobg}, the following bounds on $P(C)$ are immediate.
\begin{theorem}[\ref{subsec:thm:u-cobg}]
\label{thm:u-cobg}
    Consider the U-COBG $(\lambdavec,Y,\Sigma_{\Uup,\bg},\pi_{\Uup,\bg})$ with hashrate distribution $\lambdavec=(\lambda_j)_{j=0}^m$ and $Y=\qty{(\tx{1}{},0,f_1),(\tx{2}{},0,f_2)}$, such that $\tx{2}{}$ can be included at or after round $T>0$, and where a bribe budget $F_2:=\kappa f_1$ is distributed by $f_2$ ($\kappa>1$). Then (i) $P(C)\leq(1-\lambda_0)^T$, (ii) if $\kappa<(1-\lambda_0)^{-T+1}$, $P(C)=0$.
\end{theorem}
 We provide an implementation\footref{fot:link} to compute, based on \eqref{eq:u-cobg}, the censoring probability for a given $\lambdavec$, $f_1$ and $f_2$, where we can handle any $f_2$ written as a vector in $\R_{\geq0}^{T+1}$, where the $t$-th entry is the fee given to the miner censoring $\tx{1}{}$ at round $t-1$. For $\lambdavec'$ with $T=8$, this allows us to compute $P(C)$ as a function of $\kappa$, shown by the dashed line in Figure~\ref{fig:p}.
 \begin{figure}[h]
    \centering
    \includegraphics[width=0.8\linewidth]{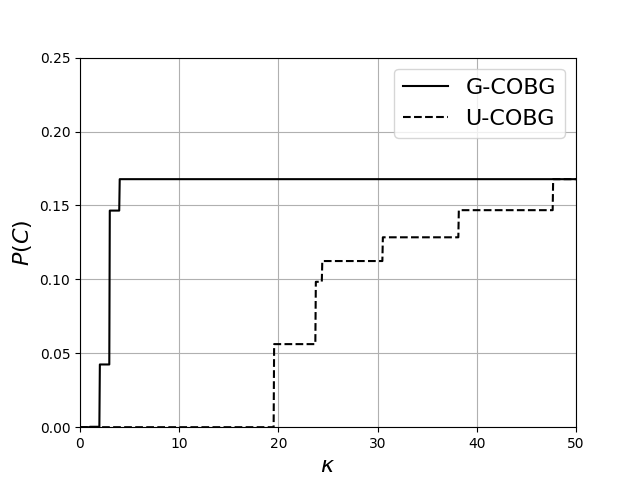}
    \caption{Censoring probability in G/U-COBG as a function of $\kappa$ for $\lambdavec'=\frac{1}{10}(2,1,3,4)$ and $T=8$.}
    \label{fig:p}
\end{figure}

\smallskip\noindent\textbf{Comparing G-COBG and U-COBG.} By Theorem~\ref{thm:g-cobg} and Theorem~\ref{thm:u-cobg}, for any $\lambda_0>0$ and number of miners $m$, we can find a lower bound on the timelock for which U-COBG is less vulnerable to timelock bribing than G-COBG, noticing that if $\kappa>m$ or $\kappa>T$, each miner until round $T$ can be bribed and so $P(C)=(1-\lambda_0)^T$ for $\kappa>m\wedge T$.
\begin{corollary}
    If $m\wedge T<(1-\lambda_0)^{-T+1}-1$, 
    the censoring probability in U-COBG is at most the censoring probability in G-COBG for any bribe budget $F_2$.
\end{corollary}
For our toy example, $m\wedge T=3<3.76\approx(1-\lambda_0)^{-T+1}-1$, in accordance with the dashed line never coming above the solid line in Figure~\ref{fig:p}. For a more realistic setting, we can look at the Bitcoin hashrate distribution of the past year\footnote{\url{https://mempool.space/graphs/mining/pools\#1y} accessed 09-10-2025.}, crudely estimating $\lambda_0=0.02$ ($\lambda_0=0.005$) as the percentage of hashrate labelled as ``Other'' (``Unknown'') by the explorer. Our reasoning here is that these miners have such a negligibly small percentage of hashrate, that they behave myopically. For $\lambda_0=0.02$, this bound guarantees that censoring is less likely to occur on Bitcoin under this model, than on an equivalent, globally predictable chain for timelocks of at least $139$ blocks. For context, the Lightning Network~\cite{poon2016bitcoin} uses timelocks of a day ($T=144$). This hints towards Lightning being more robust on Bitcoin then on a globally predictable chain, regardless of the bribe.

\smallskip\noindent\textbf{Pre-BitVM Bitcoin timelock bribing.} For specific cases, a recursion simpler than \eqref{eq:u-cobg} may hold. For example, in the case considered in \cite{nadahalli2021timelocked} that $f_2$ simply gives $F_2$ to the miner who includes $\tx{2}{}$ in round $T$ (which would be the only way to do timelock bribing in Bitcoin, before BitVM~\cite{aumayr2024bitvm,linus2024bitvm2}), we show in Theorem~\ref{thm:tau} that any miner $j=1,\ldots,m$ will try to include $\tx{1}{}$ up to some round $t_j^*$, and try to censor $\tx{1}{}$ afterwards. We henceforth assume this simplified setting.
\begin{theorem}[\ref{subsec:thm:tau}]
\label{thm:tau}
    Consider the U-COBG $(\lambdavec,Y,\allowbreak\Sigma_{\Uup,\bg},\pi_{\Uup,\bg})$ with hashrate distribution $\lambdavec=\qty(\lambda_j)_{j=0}^m$, and $Y=\qty{(\tx{1}{},0,f_1),(\tx{2}{},0,f_2)}$, such that $\tx{2}{}$ can be included at or after round $T>0$. Assume that $f_1,f_2\in\R_{\geq0}$ with $f_1<f_2$. Then the optimal strategy $(\overline{\sigma}_{\bg,j}^t)_{t=0}^{T}$ for miner $j=1,\ldots,m$ is to include $\tx{1}{}$ if $t< t_j^*$, and to censor $\tx{1}{}$ until $T$ if $t\geq t_j^*$, for any $t=0,\ldots,T$, where we let $t_j^*=T-\ceil{\rho_j}$, with $\rho_j(f_1,f_2;\lambdavec)$ defined recursively as
    \begin{footnotesize}
    \begin{equation}
    \label{eq:rho}
        \rho_j=\ceil{\rho_{j-1}}+\frac{\log\frac{f_1}{\lambda_jf_2}-\sum_{i=\ell+1}^{j-1}\qty(\ceil{\rho_{i}}-\ceil{\rho_{i-1}})\log\qty(\sum_{k=i}^m\lambda_k)}{\log\sum_{k=j}^m\lambda_k},
    \end{equation}
    \end{footnotesize}
    for $j=\ell+1,\ldots,m$, where $\ell$ is defined such that $\lambda_\ell\leq f_1/f_2<\lambda_{\ell+1}$. For $j=0,\ldots,\ell$, we have $\rho_j=0$.
    
    Moreover, the probability $P(C)$ that $\tx{2}{}$ is included is 0 if $T>\ceil{\rho_m}$ and for $\ceil{\rho_{j-1}}<T\leq \ceil{\rho_j}$ is given by
    \begin{footnotesize}
     \begin{equation*}
        P(C)=
            \qty(\sum_{k=j}^m\lambda_k)^{T-\ceil{\rho_{j-1}}}\prod_{i=1}^{j-1}\qty(\qty(\sum_{k=i}^m\lambda_k)^{\ceil{\rho_i}-\ceil{\rho_{i-1}}}).
    \end{equation*}
    \end{footnotesize}
\end{theorem} 
This result follows from assuming $\tx{1}{}$ was censored until round $t$, and then computing the expected gain for each miner in that round, starting at $t=T$ and ending at $t=0$. We then show that as $t$ decreases, more and more miners will no longer be censoring $\tx{1}{}$. The censoring probability is then computed by counting the proportion of censoring hashrate in each round until $T$. For our toy hashrate distribution $\lambdavec'$, and for any $t\geq0$, we can easily compute using \eqref{eq:rho} whether a miner would censor, were we to play a U-COBG with $T=t$, fixing $f_2=\kappa f_1$. This is done in Figure~\ref{fig:timeline} for $\kappa=25$. For $T=8$, the censoring probability is then $0.7^3\cdot0.8^5\approx0.11$, exactly the value of the dashed line at $\kappa=25$ in Figure~\ref{fig:p}.

\begin{figure}[h]
    \centering
    \includegraphics[width=\linewidth]{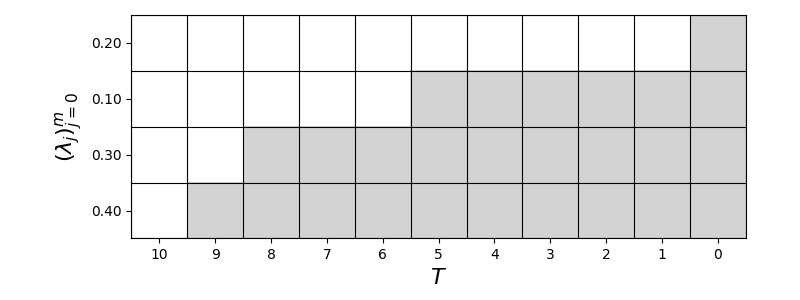}
    \caption{If the square $(\lambda,T)$ is grey, the miner with hashrate $\lambda$ censors at the start of U-COBG with timelock $T$, given $\lambdavec'=\frac{1}{10}(2,1,3,4)$ and $\kappa=25$.}
    \label{fig:timeline}
\end{figure}
For a given hashrate distribution, and fixed $f_1,f_2$, Theorem~\ref{thm:tau} tells us how large we should set the corresponding timelock $T^*$ for the censoring probability to be zero. We argue in \CameraOnly{the full version}\FullOnly{Appendix~\ref{app:worst-case}} that for given $f_1,f_2,\lambda_0$, the hashrate distribution $\lambdavec=(\lambda_0,\frac{1-\lambda_0}{2},\frac{1-\lambda_0}{2})$ yields the largest $T^*$, assuming $\lambdavec$ does not contain a majority miner (which would break consensus).  It is then easy to compute that Bitcoin timelocks of $T=144$ blocks (as in Lightning) are not vulnerable to timelock bribing as long as $f_2<\alpha f_1$, i.e., the bribe $f_2$ is smaller than $\alpha$ times the fee $f_1$ of the vulnerable transaction, with $\alpha\approx36$ ($\alpha\approx4$).

\subsection{Payment Channels}
\label{subsec:pc}
The general analysis from the previous section can be applied by our framework to any TBS application when assuming a COBG. We therefore assume pre-BitVM timelock bribing to study the security of several Hashed Timelock Contract (HTLC) and PC protocols, which are well-known examples of TBS applications.

\smallskip\noindent\textbf{An HTLC in a PC
.} 
Assume Alice and Bob share a PC with balances $v_A$ and $v_B$, and an HTLC of value $v$ from Alice to Bob with timelock $T$ and secret $\varsigma$. This state can be closed with a transaction $\tx{}{H}$. The HTLC is an output of $\tx{}{H}$ spendable in two ways: Bob claims $v$ by posting $\tx{B}{}$ and revealing $\varsigma$ in the witness, or, if Bob does not do so within $T$ blocks, Alice reclaims $v$ by posting $\tx{A}{}$. Alice and Bob can collaboratively update the channel state or unilaterally close it using a fully signed commitment transaction. We model this with a simplified application game (although it captures the IPB): if they collaborate, they either revert the payment, using $\tx{}{R}$ where Alice has $v_A+v$ and Bob $v_B$, or complete it, using $\tx{}{P}$ where Alice has $v_A$ and Bob $v_B+v$. 
\begin{definition}[HTLC Application Game]
\label{def:htlc-game}
    We define the \emph{HTLC application game $\htlc^p$ with parameter $p=(t^\varsigma,t^e)$, timelock $T$, secret $\varsigma$ (which Bob learns at time $t^\varsigma$), and values $(v_A,v_B,x)$ between $A$ and $B$} as an application game with player set $N=\qty{A,B}$, and where $(X,\Sigma_\ag,\pi)$ are defined implicitly by the sequence $(H_t)_{t=0}^{t^e}$, where $H_t$ is defined by Figures~\ref{fig:htlc-case1} and \ref{fig:htlc-case2}.
    \begin{figure*}
        \begin{subfigure}[b]{0.3\linewidth}
            \includegraphics[width=\linewidth]{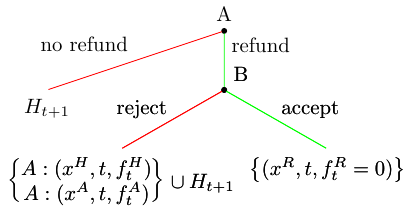}
            \caption{Tree for $H_t$ for $t<t^\varsigma$. If $t=t^e$, replace $H_{t+1}$ by $\varnothing$.}
            \label{fig:htlc-case1}
        \end{subfigure}
        \begin{subfigure}[b]{0.65\linewidth}
            \includegraphics[width=\linewidth]{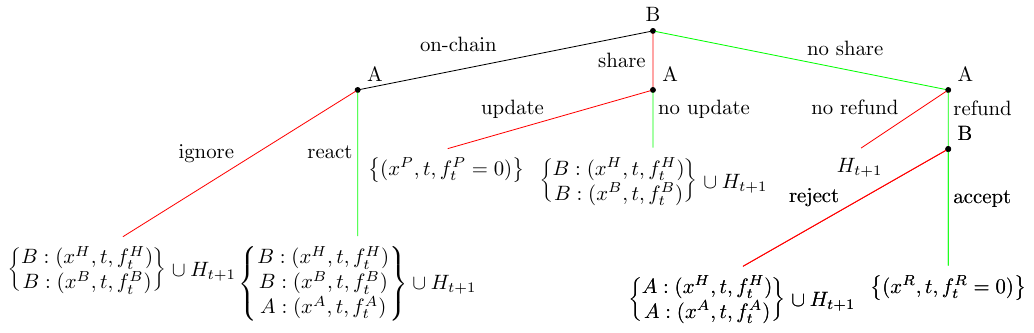}
            \caption{Tree for $H_t$ for $t\geq t^\varsigma$. If $t=t^e$, replace $H_{t+1}$ by $\varnothing$.}
            \label{fig:htlc-case2}
        \end{subfigure}
    \end{figure*}
    By $Y\cup H_{t+1}$ we mean that the game proceeds to $H_{t+1}$, with the set $Y$ added to every leaf. If a transaction triple is prefixed by ``$A:$'' or ``$B:$'', Alice or Bob posts that transaction and chooses the transaction fee at that node. If no prefix is present, the fee is zero, so the transaction is only included if no conflicting one is posted, modelling the channel remaining open with no fees paid. We indicate in red the IPB for $t<T$, which consists of waiting until the secret is known and then Bob revealing it, and in green the IPB for $t\geq T$, which consists of reverting the payment off-chain after the timelock expires.
\end{definition}
Note that setting $\P=\overline{\N}_0\times\N_0$ defines a parametrised application game $(\htlc^p)_{p\in\P}$. Each $\htlc^p$ is essentially a game tree with at each leaf a set of transaction triples. By playing the completed HTLC application game, implied by some $\beta$ and $\omega$, we can replace these sets of transaction triples by utilities for Alice and Bob. Standard computational tools could be used to solve this game, especially for larger $t^e$.

We now focus on $\htlc^p$ with $p=(0,T)$ and $T>0$, as this already captures most of the interesting behaviour\footnote{Indeed, at time $t=T$, any strategy leading to $H_{T+1}$ is strictly dominated by an action in $H_T$, where Alice will attempt to revert the payment or include $\tx{}{A}$. Moreover, a game with $t^\varsigma=t'>0$ and timelock $T>t'$ is equivalent to one with $t^\varsigma=0$ and timelock $T'=T-t'$. For $T=0$, the IPB in Definition~\ref{def:htlc-game} is not IC w.r.t.\ $\beta_0$, as Bob is indifferent between collaboratively updating the channel and closing it, paying at most $v$ in fees to get $x^B$ included (with probability $1/2$) if Alice also pays $v$ in fees to include $x^A$. Finally, for $t^\varsigma>T$, Bob never learns the secret before the game ends. The indicated IPB is again not IC, as Bob is indifferent between accepting and rejecting in round $t=T$.}.
For $\htlc^{(0,T)}$, the IPB from Definition~\ref{def:htlc-game} is not IC with respect to $\beta_0$, where $\beta_0$ is modelled by U-COBG with hashrate distribution $\lambdavec$. To make it IC, the red actions must be modified starting at $t=t^*_m(\overline{f^B},v;\lambdavec,T)$, where $t^*_m$ is defined as in Theorem~\ref{thm:tau} and $\overline{f^B}$ is the maximum fee Bob is willing to pay for $\tx{}{B}$. From this point on, Alice can force a channel closure and, with nonzero probability, censor $\tx{}{B}$, either requiring Bob to pay more than $\overline{f^B}$ to get $\tx{}{B}$ confirmed (at no loss for Alice) or censoring it entirely and profiting herself, yielding strictly higher expected utility for Alice. We do not explicitly describe the IC strategy profile.


Conversely, for any $\overline{f^B}$, Bob can choose $T>0$ such that $t^*_m(\overline{f^B},v;\lambdavec,T)>0$. Then play follows the $(\text{share,update})$ path in $H_0$, since deviation is irrational for either party\footnote{Assuming Alice prefers to keep the channel open; otherwise she is indifferent and may let Bob close and pay the fees.}. Thus, for suitable $T$ and $\overline{f^B}$, the IPB leads to $(\text{share,update})$ in $H_0$. For example, with hashrate distribution $(\lambda_0,\frac{1-\lambda_0}{2},\frac{1-\lambda_0}{2})$ and $\lambda_0=0.02$, Theorem~\ref{thm:tau} implies that for a standard Lightning HTLC timelock $T=144$, the HTLC value $v$ is at most about $36$ times $\overline{f^B}$. In other words, if we assume a payment channel holding $1\bitcoin$ between Alice and Bob, where Alice first holds $1\bitcoin$ and then transfers it to Bob, Bob can be sure he will receive his $1\bitcoin$ in case of a malicious closure if he is willing to pay $0.028\bitcoin$ in fees.


\smallskip\noindent\textbf{Two PCs
.} 
The HTLC within a PC is arguably among the simplest protocols in the blockchain setting. As we just saw, even this basic construction can fail to be IC. Nonetheless, under certain conditions, Alice and Bob may still adhere---at least partially---to the IPB. We now extend our analysis to a setting with two PCs, each containing an HTLC. While incentive compatibility under $\gvec$-composition is clearly unattainable, this example offers insight into how compositional effects can reshape incentive structures.

To this end, consider two PC: one between Alice and Bob, and one between Charlie and Dave. Both channels set up HTLCs at $t=0$, with timelocks $T_1,T_2$, secrets $\varsigma_1,\varsigma_2$, and values $(v_A,v_B,v_1),(v_C,v_D,v_2)$, respectively. We model both PCs by parametrised HTLC application games $(\htlc_1^{(t^\varsigma_1,t^e_1)})_{(t_1^\varsigma,t_1^e)\in\P}$ and $(\htlc_2^{(t^\varsigma_2,t^e_2)})_{(t_2^\varsigma,t_2^e)\in\Q}$ with $\P=\Q=\overline{\N}_0\times\N_0$. Using suitable compositions, we can analyse cases where one channel influences the other. For simplicity, assume $t_1^e=t_2^e=T:=T_1\vee T_2$. As in the single-channel case, this only removes dominated strategies.

We now study, w.l.o.g., two compositions of $(\htlc^{(p,T)}_1)_{p\in\overline{\N}_0}$ and $(\htlc^{(q,T)}_2)_{q\in\overline{\N}_0}$. Trivially, we could define the collection $\gvec^{ind}=(g_q^{ind})_{q\in\overline{\N}_0}$ where for all $q\in\overline{\N}_0$, $g_q^{ind}(\sigma_{\ag,2})=q$ for every $\sigma_{\ag,2}\in\Sigma_{\ag,2}^{(q,T)}$. This ``independent'' composition models two unrelated channels, where Bob and Dave learn $\varsigma_1$ and $\varsigma_2$ simultaneously at time $q\in\overline{\N}_0$. Note that, had we defined the ``HTLC in PC''-IPB differently, Theorem~\ref{thm:constant-composition} would imply security of the modified protocols under $\gvec^{ind}$-composition.

Another collection is $\gvec^{dep}=(g_q^{dep})_{q\in\overline{\N}_0}$ where for all $q\in\overline{\N}_0$, $g_q^{dep}(\sigma_{\ag,2})=\tau_C(\sigma_{\ag,2})$ for every $\sigma_{\ag,2}\in\Sigma_{\ag,2}^{(q,T)}$, where $\tau_C(\sigma_{\ag,2})\in\overline{\N}_0$ is the time Charlie learns the secret in case strategy profile $\sigma_{\ag,2}\in\Sigma_{\ag,2}^{(q,T)}$ is chosen. This ``dependent'' composition models the case where Charlie and Bob are the same entity, or where Charlie immediately shares the secret with Bob. Whether $\varsigma_1=\varsigma_2$ or $\varsigma_1\neq \varsigma_2$ is irrelevant, as the composition represents Bob figuring out $\varsigma_1$ as soon as Charlie learns $\varsigma_2$.

Extra requirements are imposed for security under $\gvec^{dep}$-composition. In particular, consider a CR where Alice and Dave form one party, and Bob and Charlie another, yielding channels Dave–Charlie ($\htlc_1$) and Charlie–Dave ($\htlc_2$). Then $T_1$ must be sufficiently larger than $T_2$ to prevent Dave from deviating and causing Charlie to lose funds.
\begin{theorem}[\ref{subsec:thm:2htlc-timelock}]
\label{thm:2htlc-timelock}
    In the above setting, with the usual U-COBG, it is necessary in order for Dave not to deviate from the IPB, that $t^*_m(\overline{f^C_1},v_1;\lambdavec,T_1)+1\geq t^*_m(\overline{f^D_2},v_2;\lambdavec,T_2)$.
\end{theorem}
With $\lambdavec=(\lambda_0,\frac{1-\lambda_0}{2},\frac{1-\lambda_0}{2})$ and $\lambda_0=0.02$, if $T_1=288$, $T_2=144$, $v_1=v_2=:v$, if Charlie is only willing to pay up to $\frac{1}{60}v$ in transaction fees, whereas Dave is willing to pay up to $\frac{1}{3}v$, we have $t^*_m(\frac{1}{60}v,v;\lambdavec,288)+1=119<124=t^*_m(\frac{1}{3}v,v;\lambdavec,144)$, allowing Dave to only release the secret at $t=119$, updating the Charlie-Dave channel off-chain, while having a positive probability to censor Charlie trying to claim the HTLC in the Dave-Charlie channel. 

\smallskip\noindent\textbf{Three PCs
.} The parametrised game approach easily scales to multiple compositions. For example, consider three channels Alice–Bob, Bob–Charlie, and Charlie–Dave, each setting up an HTLC at $t=0$ with timelocks $T_1,T_2,T_3$, secrets $\varsigma_1,\varsigma_2,\varsigma_3$, and values $(v^A_1,v^B_1,v_1)$, $(v^B_2,v^C_2,v_2)$, and $(v^C_3,v^D_3,v_3)$. Assume the timelocks satisfy the condition of Theorem~\ref{thm:2htlc-timelock}, mutatis mutandis. The PCs are modelled by parametrised HTLC application games $(\htlc_1^{(t_1^\varsigma,T)})_{t_1^\varsigma\in\overline{\N}_0}$, $(\htlc_2^{(t_2^\varsigma,T)})_{t_2^\varsigma\in\overline{\N}_0}$ and $(\htlc_3^{(t_3^\varsigma,T)})_{t_3^\varsigma\in\overline{\N}_0}$, where $T:=T_1\vee T_2\vee T_3$. Using suitable compositions, we can model Bob–Dave collusion so that Bob learns $\varsigma_1$ when Dave learns $\varsigma_3$. Specifically, we can study the composition $\htlc:=\htlc_1\circ_{\gvec^{12}}\htlc_2\circ_{\gvec^{23}}\htlc_3$, where $\gvec^{23}$ would be like $\gvec^{dep}$, i.e., $\gvec^{23}=(g_r^{23})_{r\in\overline{\N}_0}$ with for all $r\in\overline{\N}_0$, $g^{23}_r(\sigma_{\ag,3})=\tau_D(\sigma_{\ag,3})$
, and where $\gvec^{12}$ would be like $\gvec^{ind}$, i.e., $\gvec^{12}=(g^{12}_q)_{q\in\overline{\N}_0}$ with for all $q\in\overline{\N}_0$, $g^{12}_q(\sigma_{\ag,2})=q$ 
. From now on, we use ``Dave'' to denote both Bob and Dave.
\begin{theorem}[\ref{subsec:thm:wormhole}]
\label{thm:wormhole}
    In the above setting with U-COBG, Dave deviates from the IPB in $\htlc_1\circ_{\gvec^{12}}\htlc_2\circ_{\gvec^{23}}\htlc_3$ if $v_2>v_3$.
\end{theorem}
In essence, Theorem~\ref{thm:wormhole} shows that Dave can \emph{steal} the routing fee intended for Charlie, since typically $v_2=v_3+f_{route}^C>v_3$ to compensate Charlie for routing the payment. Charlie believes the payment failed, while Dave captures $f_{route}^C$ by sharing the secret directly to Alice (via Bob). This demonstrates that the HTLC construction used in systems such as the Lightning Network is vulnerable to the \emph{wormhole attack} \cite{malavolta2018anonymous}. Known mitigations would, in our framework, rule out the $\gvec^{12}$-composition and instead enforce a dependent composition of $\htlc_1$ and $\htlc_2$, similar to $\gvec^{23}$.

\smallskip\noindent\textbf{CRAB PCs
.} Apart from the timelock bribing on HTLCs described in this section, one can mount bribing attacks to post outdated channel commitment transactions and censor the ensuing punishment transaction. CRAB PCs \cite{aumayr2024securing} do not have this vulnerability as participants must put in a collateral $c$ given to the miner by the punishment transaction. \CameraOnly{As can be seen in the full version,}\FullOnly{As can be seen in Appendix~\ref{app:timelock-examples},} we conclude just like \cite{aumayr2024securing} that setting $c>\frac{v}{2}$, with $v$ the channel capacity, ensures rational participants will not post old commitment transactions.

\FullOnly{
    \subsection{Sandwich Attack}
    \label{subsec:mev}
    Consider a DeFi smart contract where a user $U$ places a buy limit order $\tx{U}{}$ at price $l$, while the current price is $l - p$ for some $p \geq 0$. If included directly, $U$ gains utility $p$. However, a miner $j$ can front-run with $\tx{F,j}{}$ (buying at price $l - p$), and back-run with $\tx{B,j}{}$ (selling at price $l$), capturing the profit $p$ and leaving $U$ with zero utility.
    Until now, the network game has been kept intentionally simple: all players in the blockchain game receive the full set of transaction triples produced by the application game. A richer network game could now be introduced in which $U$ decides to which miners to send $\tx{U}{}$. This models limited propagation which can influence strategic behaviour. Although miners might already know $U$ wants to make the trade $\tx{U}{}$, they would only be willing to sandwich $\tx{U}{}$ if they have actually received $\tx{U}{}$ as input to the subsequent blockchain game. A more detailed analysis in Appendix~\ref{app:mev} shows how $U$ remains vulnerable to a sandwich attack, unless we introduce a trusted miner who would only include $\tx{U}{}$. Future work could explore how repeated interactions in the presence of such a trusted miner may incentivise rational miners to not perform sandwich attacks in order to gain future fees of transactions that a user would otherwise not choose to share with that miner.

}

\CameraOnly{
    \subsection{Extensions and Broader Applications}
    \label{subsec:general-examples}
    Our framework is designed to facilitate compositional reasoning in complex blockchain environments where incentive alignment arises from the interaction of multiple protocols and layers. While earlier case studies have formalised specific constructions, the framework naturally generalises to a broader class of real-world settings. This subsection first briefly describes how modelling the network layer helps capturing these real-world settings. Then, we outline three principal categories that reflect common compositional patterns in practice, each introducing distinct modelling requirements and incentive dynamics. Detailed examples are provided in the full version. 
    
    \smallskip\noindent\textit{The network game.} The application game outputs a set of transaction triples, but these do not immediately become inputs to the blockchain game. In practice, transactions are relayed over a communication network, whose structure determines how and when they are observed. Mempool visibility may vary and adversarial actors can selectively delay or withhold propagation, leading to diverging views and strategic behaviour amongst players. We can capture the features of the network in the \emph{network game}. Through this game, both application game players and validators can strategise to determine with which transactions each validator starts the blockchain game. In settings where the ordering of transactions within blocks in the blockchain might be relevant, such as Maximal Extractable Value (MEV) \cite{daian2020flash}, a flexible network game is crucial. It allows validators to exploit their privileged position to determine the transaction ordering: potentially inserting, reordering, and leaving out transactions given the actions from the application layer. For a more detailed account on the network game, we refer to the full version.\pim{please check} 
}

\FullOnly{
    \subsection{Broader Applications}
    \label{subsec:general-examples}
    Our framework is designed to facilitate compositional reasoning in complex blockchain environments where incentive alignment arises from the interaction of multiple protocols and layers. While earlier case studies have formalised specific constructions, the framework naturally generalises to a broader class of real-world settings. This subsection outlines three principal categories that reflect common compositional patterns in practice, each introducing distinct modelling requirements and incentive dynamics. Detailed examples are provided in Appendix~\ref{app:examples}. 
}

\smallskip\noindent\textit{Cross-application composition.} Multiple decentralised applications (dApps) share a common blockchain infrastructure and may interact implicitly through shared state or timing dependencies. Representative examples include arbitrage between decentralised exchanges (DEXs) or oracle-driven feedback loops, where strategic behaviour emerges only from the interplay of otherwise independent systems.

\smallskip\noindent\textit{Cross-blockchain composition.} Protocols such as atomic swaps coordinate actions across independent blockchains with distinct execution environments and trust models. By treating each chain as a separate game and specifying inter-chain dependencies explicitly, our framework enables modular analysis under heterogeneous assumptions.

\smallskip\noindent\textit{Complex Network and Multi-layer Dynamics.}\pim{Change?} Architectures like proposer-builder separation (PBS) span the application, network, and consensus layers. Our model treats each layer explicitly, supporting rigorous reasoning about how network-level logic (e.g., auctions among builders) and consensus-level inclusion policies jointly affect incentives. 

These categories illustrate how the framework accommodates modern protocol designs that defy traditional single-layer analysis. Concrete examples corresponding to each setting are outlined informally in \CameraOnly{the full version}\FullOnly{Appendix~\ref{app:examples}}, 
focusing on how the framework could be instantiated to capture the relevant dynamics. Although these use cases are not fully formalised, they demonstrate how modular reasoning can be leveraged to isolate critical assumptions, detect edge-case vulnerabilities, and explore design alternatives.

Crucially, each category supports meaningful questions that would be difficult to approach without a layered and compositional model. In cross-application scenarios, such a model enables formal reasoning about whether composability introduces profitable but unintended deviations, e.g., does arbitrage across DEXs destabilise pricing mechanisms, or can oracle design be hardened against feedback-driven manipulation? 
In cross-chain protocols, our framework provides the structure to analyse swap soundness under heterogeneous security assumptions, e.g., how do differences in block times, censorship resistance, or finality affect incentive compatibility?
Finally, in multi-layer systems such as PBS and MEV auctions, the framework enables principled mechanism design: which auction formats discourage censorship? When does exclusive order flow lead to centralisation? How should fees and rebates be structured to align incentives across builders, searchers, and proposers? 
These examples demonstrate how the compositional design of our framework provides a modular and rigorous foundation for evaluating incentive properties in emerging blockchain protocols.

\section{Related Work}
\label{subsec:related-work}

\textit{Layer-specific models.}  
A large body of work studies incentive mechanisms within isolated layers of the blockchain stack. 
At the consensus layer, researchers have analysed block rewards~\cite{eyal2018majority,budish2024economic,chen2019axiomatic} and transaction fee mechanisms~\cite{bahrani2024transactionfc,bahrani2024transaction,roughgarden2024transaction}. Other works model consensus dynamics using rational agents~\cite{abraham2013distributed, aiyer2005bar, amoussou2019rationals, kiayias2016blockchain, bentov2021tortoise, kiayias2021coalition, liao2017incentivizing, winzer2019temporary, liu2019survey, pass2017fruitchains, budish2024economic}. Similarly, other works restrict their analysis on the network layer, e.g.~\cite{babaioff2012bitcoin}. Several works conduct a rational analysis on the application layer, either targeting specific applications like auctions~\cite{chitra2023credible} or off-chain systems like payment channels~\cite{avarikioti2020ride,avarikioti2023muskateer,avarikioti2020cerberus,rain2021towards,aumayr2025xtrasnfer}.
However, these results assume fixed behaviour from other layers and do not account for strategic interactions across layers or applications. In contrast, our framework explicitly models such cross-layer and cross-application dependencies through the cross-layer game and the composition of application games. 

\textit{Monolithic protocol analyses.}  
Some works jointly analyse protocol layers, especially in Layer-2 constructions where application behaviour and miner incentives interact~\cite{tsabary2021mad, wadhwa2022helium,aumayr2024securing,avarikioti2022suborn,avarikioti2019brick,avarikioti2025thunderdome,doe2023incentive}. For instance,~\cite{tsabary2021mad,wadhwa2022helium} model how rational miners affect the execution of HTLCs. However, these analyses are monolithic, capturing the entire system in a single game. This limits extensibility and compositional reasoning: security properties must be re-derived from scratch for each new protocol interaction. Our approach instead defines modular game components---called games---with formally specified interfaces, allowing reuse and composition across protocols.

\textit{MEV and dynamic incentives.}  
Daian et al.~\cite{daian2020flash} introduced the idea of Maximal Extractable Value (MEV), showing how application-level transactions affect miner strategies. Follow-up work has proposed incentive-aligned auctions and mitigations~\cite{doe2023incentive}, but remains tied to specific systems and does not generalise to arbitrary protocol interactions. \CameraOnly{The full version of our}\FullOnly{Our} framework abstracts MEV via the network game, positioned between the application and blockchain layers, enabling systematic analysis across diverse protocol designs.


\textit{Cross-application and cross-layer security.}  
More closely related to our work is Zappala et al.~\cite{zappala2021framework}, who aim to formalise secure protocol composition. Unlike our framework, theirs assumes subgame independence rather than deriving it, limiting applicability to interdependent protocols. For instance, in their Lightning Network~\cite{poon2016bitcoin} analysis, HTLCs are treated as independent if posted on time, avoiding strategic deviations. Their security conditions omit dynamic strategic behaviours, whereas our framework supports compositional reasoning without assuming independence. This enables analysis of incentive-driven interactions such as adversarial timing and MEV extraction across concurrent protocols.

\textit{Compositional game theory.}  
Our work is inspired by compositional game theory (CGT)~\cite{ghani2016compositional}, which uses category theory to model complex systems as compositions of modular components called open games. The games defined in Section~\ref{sec:model} can each be translated into open games. For instance, a parametrised application game $(\G^p)_{p\in\P}=(N^p,X^p,\Sigma^p,\pi^p)$ corresponds to an open game $\G:(\P,\qty{*})\to(Z,R)$, where $Z = \qty{\pi^p(\sigma^p):p\in\P,\sigma^p\in\Sigma^p}$ and $R = \R$. This open game is a 4-tuple $\G=(\Sigma,f_P,f_C,f_B)$, where $\Sigma = \prod_{p\in\P}\Sigma^p$, $f_P:\Sigma\times\P\to Z$ is defined by $f((\sigma^q)_{q\in\P},p):=\pi^p(\sigma^p)$, and $f_C:\Sigma\times\P\times R\to\qty{*}$ is trivial since application games offer no coutility. The best response function $f_B:\P\times(Z\to R)\to\text{Rel}(\Sigma)$ maps $(p,u)$ to the best response relation on $\Sigma\times\Sigma$ (intuitively: $(\sigma,\sigma)\in f_B(p,u)$ if $\sigma$ is a Nash equilibrium). The utility function $u:Z\to R$ represents the utility derived from the application outcome, based on the subsequent network and blockchain (open) games and the execution function.


This approach provides a rigorous and elegant language for composing games, but remains largely theoretical, offering few concrete results on how strategic behaviour composes. In contrast, our framework instantiates the principles of CGT in the blockchain setting, leveraging the layered architecture to define explicit interfaces (e.g., transaction propagation and fee mechanisms) and enabling strategic reasoning across interdependent protocols. This yields concrete results on properties of specific (classes of) protocols, and how they change under composition, making the framework practical and promoting broader adoption of game-theoretic security reasoning in blockchain design.


\textit{Rational protocol design and simulation-based models.}
Our approach differs fundamentally from Rational Protocol Design (RPD)~\cite{garay2013rational,badertscher2018but}, which models security as a Stackelberg game between a protocol designer and an attacker, with utilities defined via ideal functionalities in the simulation paradigm. RPD's “composition” theorem refers to subroutine replacement, enabling game-theoretic analysis in a hybrid model where cryptographic components are idealised. The theorem ensures RPD-security when these are instantiated by implementations satisfying cryptographic (i.e. non-game-theoretic) security. This addresses issues in earlier rational models where, for example, assuming the existence of a secure communication channel turned out to not preserve equilibria when instantiated with a public-key infrastructure on insecure channels. In contrast, we view composition as building a rational protocol from rational parts, composing them cross-layer and cross-application, enabling reasoning about equilibria in composed games.

\textit{Automated tools.}
Recent automated tools for reasoning about rational security, such as CheckMate~\cite{brugger2023checkmate}, focus on verifying strategic properties in isolated protocols using symbolic game representations. These tools provide valuable automation but do not model layered architectures or inter-protocol interactions. In contrast, our framework is able to analyse systems spanning multiple layers and applications.

\textit{This work.}  
While prior works offer valuable insights into rational security, they either focus narrowly on isolated components, assume idealised independence between interacting parts, lack a modular framework for capturing the dynamic nature of incentive interactions across protocols and layers, or are too abstract to yield concrete results in blockchain settings. Our framework fills this gap by enabling modular, game-theoretic analysis of blockchain systems. 

\section{Conclusion}
\label{sec:conclusion}
In this work, we presented a compositional framework for analysing the game-theoretic security of blockchain protocols. Unlike traditional approaches that analyse individual protocols in isolation or assume fixed-layer behaviour, our model embraces the modular structure of modern blockchain ecosystems by decomposing them into layered games. Each layer---application, blockchain, and potentially network---is formalised as a strategic game, and interactions are expressed through explicitly defined interfaces. This layered architecture allows us to reason not only about each component's behaviour but also about how incentives propagate across the system.

At the core of our approach lie the definitions of cross-layer games and cross-application composition, which enable reasoning about concurrent protocol execution and emerging vulnerabilities such as censorship or MEV. These abstractions support modular security analysis: properties of individual layers can be verified and then composed, facilitating the study of complex interactions without rederiving the entire system’s behaviour. Our use of parametrised games further introduces modelling flexibility, allowing protocols to interleave in time or through shared components.

Through several detailed case studies, we demonstrate the expressiveness and utility of our framework. These examples show how existing incentive misalignments can be captured and formalised, revealing vulnerabilities and new levers for robust protocol design.

This work also opens several directions for future research. First, one should determine the possibility of fully separating the analysis of application and blockchain layers without limiting the framework applicability to a sequential composition of application and blockchain games. Second, extending the framework further to capture richer network models, including asynchronous message propagation, selective delivery, or targeted censorship, would reflect the operational characteristics of blockchain networks even more faithfully. 
Third, formalising the translation from protocol specifications to application games is essential to enable broader applicability and automation. Fourth, identifying classes of protocols that satisfy compositional incentive compatibility under various assumptions would clarify the framework's expressive limits.

Of separate interest is the further exploration of Censor-Only Blockchain Games; in particular the derivation of exact expressions for the censoring probability in the unpredictable and locally predictable variants, and identifying when each variant is less prone to timelock bribing. Finally, generalising the framework to support parametrised families of blockchain environments, e.g., based on hashrate distributions or latency assumptions, and exploring alternative definitions of incentive compatibility could yield practically relevant, robust security guarantees across diverse deployments. 

Our framework bridges a long-standing gap between incentive-aware protocol analysis and modular security reasoning. It provides a foundation for understanding the strategic behaviour of blockchain participants in complex, multi-protocol environments and contributes to the development of provably secure, incentive-aligned decentralised systems.

\bibliographystyle{IEEEtran}
\bibliography{references}

\appendices


\FullOnly{
    \section{Table of Most Important Notation}
\label{app:notation}

\CameraOnly{
\begin{strip}
\centering
\captionof{table}{Most important notation.}
    \label{tab:notation}
    \resizebox{\textwidth}{!}{%
\begin{tabular}{c|l|c|l|c|l}
\hline
$\X$           & \begin{tabular}[c]{@{}l@{}}Set of all possible transactions \\ in a blockchain ecosystem\end{tabular}           & $\lambdavec=(\lambda_j)_{j=0}^m$                                                                    & Hashrate distribution                                                                                                            & $\F_{\Pi,\beta,\omega}$                                  & \begin{tabular}[c]{@{}l@{}}Cross-layer game with respect \\ to $\beta$ and $\omega$ of $\Pi$\end{tabular}          \\ \hline
$\Omega$       & \begin{tabular}[c]{@{}l@{}}Set of all orderings of a (subset \\ of) transactions in $\X$\end{tabular}           & $\B_{\Uup/\Gup/\Lup}(Y)$                                                                            & \begin{tabular}[c]{@{}l@{}}Unpredictable/Globally Predictable\\ /Locally Predictable Censor-Only \\ Blockchain Game\end{tabular} & $\B(\G)$                                                 & \begin{tabular}[c]{@{}l@{}}Collection of blockchain games \\ induced by outcome of $\G$\end{tabular}               \\ \hline
$\beta$        & Block behaviour                                                                                                 & $D$                                                                                                 & Epoch length                                                                                                                     & $\Sigma_{\bg}(\sigma_{\ag}), \sigma_{\bg}(\sigma_{\ag})$ & $\Sigma_{\bg}^{\pi_\ag^p(\sigma_\ag)}, \sigma_{\bg}^{\pi_\ag^p(\sigma_\ag)}$                                       \\ \hline
$s,t$          & Time indices                                                                                                    & $e$                                                                                                 & Epoch index                                                                                                                      & $\pi_{\bg}^{(\sigma_\ag)}$                               & $\pi_{\bg}^{\pi_\ag^p(\sigma_\ag)}$                                                                                \\ \hline
$A_h$          & Block at height $h$                                                                                              & $\vvec^e=(v_t^e)_{t=0}^{D-1}$                                                                       & Miner assignment for epoch $e$                                                                                                   & $\sigma_{\bg}(\cdot)$                                    & $(\sigma_{\bg}(\sigma_{\ag}'))_{\sigma_{\ag}'\in\Sigma_{\ag}^p}$                                                   \\ \hline
$i,j$          & \begin{tabular}[c]{@{}l@{}}Indices used to represent players,\\ $j$ is usually reserved for miners\end{tabular} & $\vvec^e|_j$                                                                                        & \begin{tabular}[c]{@{}l@{}}Rounds for which $j$ is chosen \\ as leader by $\vvec^e$\end{tabular}                                 & $P$                                                      & Probability measure                                                                                                \\ \hline
$\omega$       & Execution function                                                                                              & $\mu_{j,t}$                                                                                         & \begin{tabular}[c]{@{}l@{}}Block proposal function for \\ miner $j$ at round $t$\end{tabular}                                    & $u=(u_i)_{i\in N\cup M}$                                 & Vector of utility functions for each $i$                                                                           \\ \hline
$M$            & \begin{tabular}[c]{@{}l@{}}Set of blockchain game players \\ (miners/validators)\end{tabular}                   & $b_{j,t}$                                                                                           & \begin{tabular}[c]{@{}l@{}}Block proposal of miner $j$ \\ at round $t$\end{tabular}                                              & $\eta$                                                   & Collusion reduction transformation                                                                                 \\ \hline
$\tx{}{}$      & Transaction                                                                                                     & $\Pi$                                                                                               & Application protocol                                                                                                             & $\tilde{\Ga}$                                            & Collusion reduction of $\Ga$                                                                                       \\ \hline
$f$            & Fee function                                                                                                    & $N$                                                                                                 & Application player set                                                                                                           & $\gvec=(g_p)_{p\in\P}$                                   & Collection of composition functions                                                                                \\ \hline
$\D(S)$        & \begin{tabular}[c]{@{}l@{}}Set of probability distributions \\ over $S$\end{tabular}                            & $\G$                                                                                                & Application game                                                                                                                 & $\circ_\gvec$                                            & $\gvec$-composition                                                                                                \\ \hline
$X$            & Set of transactions                                                                                             & $\Ga$                                                                                               & Arbitrary game                                                                                                                   & $\1_E$                                                   & \begin{tabular}[c]{@{}l@{}}Indicator function; equals $1$ \\ if $E$ is true and $0$ otherwise\end{tabular}         \\ \hline
$\Y(X)$        & \begin{tabular}[c]{@{}l@{}}Set of all transaction triples \\ of $X$\end{tabular}                                & \begin{tabular}[c]{@{}c@{}}$\sigma_{\bg/\ag/\cg},$\\ $\overline{\sigma}_{\bg/\ag/\cg}$\end{tabular} & \begin{tabular}[c]{@{}l@{}}strategy profile of an \\ application/cross-layer game\end{tabular}                                   & $T$                                                      & Timelock                                                                                                           \\ \hline
$\Po(S)$       & Power set of $S$                                                                                                & $\Sigma_{\bg/\ag/\cg}$                                                                              & \begin{tabular}[c]{@{}l@{}}Set of strategy profiles of a blockchain/\\ application/cross-layer game\end{tabular}                 & $F_2=\kappa f_1$                                         & \begin{tabular}[c]{@{}l@{}}Bribe budget, available to try to \\ censor a transaction paying fee $f_1$\end{tabular} \\ \hline
$Y$            & Set of transaction triples                                                                                      & $\pi_{\bg/\ag}$                                                                                     & \begin{tabular}[c]{@{}l@{}}Outcome function of a blockchain/\\ application game\end{tabular}                                     & $C$                                                      & Event that $\tx{1}{}$ gets censored                                                                                \\ \hline
$b$            & Blockchain ordering                                                                                             & $(\Ga^p)_{p\in\P}$                                                                                  & Collection of parametrised games                                                                                                 & $N_0^T$                                                  & \begin{tabular}[c]{@{}l@{}}Number of distinct validators in \\ rounds $0,\ldots,T$\end{tabular}                    \\ \hline
$\bmodel_M(Y)$ & \begin{tabular}[c]{@{}l@{}}Set of possible blockchain orderings \\ with transactions from $Y$\end{tabular}      & $p,q,\P,\Q$                                                                                         & \begin{tabular}[c]{@{}l@{}}Parameters and their respective\\ parameter spaces\end{tabular}                                       & $\vee,\wedge$                                            & Maximum, minimum                                                                                                   \\ \hline
$\B(Y)$        & Blockchain game induced by $Y$                                                                                  & $\No(Y)$                                                                                            & Network game induced by $Y$                                                                                                      & $g_j^t$                                                  & \begin{tabular}[c]{@{}l@{}}Expected gain of miner $j$ \\ at round $t$\end{tabular}                                 \\ \hline
\end{tabular}
}
\end{strip}
}

\CameraOnly{
\begin{strip}
\centering
\captionof{table}{Most important notation.}
    \label{tab:notation}
    \resizebox{0.5\textwidth}{!}{%
\begin{tabular}{c|l|c|l}
\hline
$\X$                             & \begin{tabular}[c]{@{}l@{}}Set of all possible transactions \\ in a blockchain ecosystem\end{tabular}                            & $N$                                                                                                 & Application player set                                                                                             \\ \hline
$\Omega$                         & \begin{tabular}[c]{@{}l@{}}Set of all orderings of a (subset \\ of) transactions in $\X$\end{tabular}                            & $\G$                                                                                                & Application game                                                                                                   \\ \hline
$\beta$                          & Block behaviour                                                                                                                  & $\Ga$                                                                                               & Arbitrary game                                                                                                     \\ \hline
$s,t$                            & Time indices                                                                                                                     & \begin{tabular}[c]{@{}c@{}}$\sigma_{\bg/\ag/\cg},$\\ $\overline{\sigma}_{\bg/\ag/\cg}$\end{tabular} & \begin{tabular}[c]{@{}l@{}}strategy profile of an \\ application/cross-layer game\end{tabular}                     \\ \hline
$A_h$                            & Block at height $h$                                                                                                               & $\Sigma_{\bg/\ag/\cg}$                                                                              & \begin{tabular}[c]{@{}l@{}}Set of strategy profiles of a blockchain/\\ application/cross-layer game\end{tabular}   \\ \hline
$i,j$                            & \begin{tabular}[c]{@{}l@{}}Indices used to represent players,\\ $j$ is usually reserved for miners\end{tabular}                  & $\pi_{\bg/\ag}$                                                                                     & \begin{tabular}[c]{@{}l@{}}Outcome function of a blockchain/\\ application game\end{tabular}                       \\ \hline
$\omega$                         & Execution function                                                                                                               & $(\Ga^p)_{p\in\P}$                                                                                  & Collection of parametrised games                                                                                   \\ \hline
$M$                              & \begin{tabular}[c]{@{}l@{}}Set of blockchain game players \\ (miners/validators)\end{tabular}                                    & $p,q,\P,\Q$                                                                                         & \begin{tabular}[c]{@{}l@{}}Parameters and their respective\\ parameter spaces\end{tabular}                         \\ \hline
$\tx{}{}$                        & Transaction                                                                                                                      & $\No(Y)$                                                                                            & Network game induced by $Y$                                                                                        \\ \hline
$f$                              & Fee function                                                                                                                     & $\F_{\Pi,\beta,\omega}$                                                                             & \begin{tabular}[c]{@{}l@{}}Cross-layer game with respect \\ to $\beta$ and $\omega$ of $\Pi$\end{tabular}          \\ \hline
$\D(S)$                          & \begin{tabular}[c]{@{}l@{}}Set of probability distributions \\ over $S$\end{tabular}                                             & $\B(\G)$                                                                                            & \begin{tabular}[c]{@{}l@{}}Collection of blockchain games \\ induced by outcome of $\G$\end{tabular}               \\ \hline
$X$                              & Set of transactions                                                                                                              & $\Sigma_{\bg}(\sigma_{\ag}), \sigma_{\bg}(\sigma_{\ag})$                                            & $\Sigma_{\bg}^{\pi_\ag^p(\sigma_\ag)}, \sigma_{\bg}^{\pi_\ag^p(\sigma_\ag)}$                                       \\ \hline
$\Y(X)$                          & \begin{tabular}[c]{@{}l@{}}Set of all transaction triples \\ of $X$\end{tabular}                                                 & $\pi_{\bg}^{(\sigma_\ag)}$                                                                          & $\pi_{\bg}^{\pi_\ag^p(\sigma_\ag)}$                                                                                \\ \hline
$\Po(S)$                         & Power set of $S$                                                                                                                 & $\sigma_{\bg}(\cdot)$                                                                               & $(\sigma_{\bg}(\sigma_{\ag}'))_{\sigma_{\ag}'\in\Sigma_{\ag}^p}$                                                   \\ \hline
$Y$                              & Set of transaction triples                                                                                                       & $P$                                                                                                 & Probability measure                                                                                                \\ \hline
$b$                              & Blockchain ordering                                                                                                              & $u=(u_i)_{i\in N\cup M}$                                                                            & Vector of utility functions for each $i$                                                                           \\ \hline
$\bmodel_M(Y)$                   & \begin{tabular}[c]{@{}l@{}}Set of possible blockchain orderings \\ with transactions from $Y$\end{tabular}                       & $\eta$                                                                                              & Collusion reduction transformation                                                                                 \\ \hline
$\B(Y)$                          & Blockchain game induced by $Y$                                                                                                   & $\tilde{\Ga}$                                                                                       & Collusion reduction of $\Ga$                                                                                       \\ \hline
$\lambdavec=(\lambda_j)_{j=0}^m$ & Hashrate distribution                                                                                                            & $\gvec=(g_p)_{p\in\P}$                                                                              & Collection of composition functions                                                                                \\ \hline
$\B_{\Uup/\Gup/\Lup}(Y)$         & \begin{tabular}[c]{@{}l@{}}Unpredictable/Globally Predictable\\ /Locally Predictable Censor-Only \\ Blockchain Game\end{tabular} & $\circ_\gvec$                                                                                       & $\gvec$-composition                                                                                                \\ \hline
$D$                              & Epoch length                                                                                                                     & $\1_E$                                                                                              & \begin{tabular}[c]{@{}l@{}}Indicator function; equals $1$ \\ if $E$ is true and $0$ otherwise\end{tabular}         \\ \hline
$e$                              & Epoch index                                                                                                                      & $T$                                                                                                 & Timelock                                                                                                           \\ \hline
$\vvec^e=(v_t^e)_{t=0}^{D-1}$    & Miner assignment for epoch $e$                                                                                                   & $F_2=\kappa f_1$                                                                                    & \begin{tabular}[c]{@{}l@{}}Bribe budget, available to try to \\ censor a transaction paying fee $f_1$\end{tabular} \\ \hline
$\vvec^e|_j$                     & \begin{tabular}[c]{@{}l@{}}Rounds for which $j$ is chosen \\ as leader by $\vvec^e$\end{tabular}                                 & $C$                                                                                                 & Event that $\tx{1}{}$ gets censored                                                                                \\ \hline
$\mu_{j,t}$                      & \begin{tabular}[c]{@{}l@{}}Block proposal function for \\ miner $j$ at round $t$\end{tabular}                                    & $N_0^T$                                                                                             & \begin{tabular}[c]{@{}l@{}}Number of distinct validators in \\ rounds $0,\ldots,T$\end{tabular}                    \\ \hline
$b_{j,t}$                        & \begin{tabular}[c]{@{}l@{}}Block proposal of miner $j$ \\ at round $t$\end{tabular}                                              & $\vee,\wedge$                                                                                       & Maximum, minimum                                                                                                   \\ \hline
$\Pi$                            & Application protocol                                                                                                             & $g_j^t$                                                                                             & \begin{tabular}[c]{@{}l@{}}Expected gain of miner $j$ \\ at round $t$\end{tabular}                                 \\ \hline
\end{tabular}
}
\end{strip}
}

\FullOnly{
\begin{strip}
\centering
\captionof{table}{Most important notation.}
    \label{tab:notation}
    \resizebox{\textwidth}{!}{%
\begin{tabular}{c|l|c|l}
\hline
$\X$                             & \begin{tabular}[c]{@{}l@{}}Set of all possible transactions \\ in a blockchain ecosystem\end{tabular}                            & $\G$                                                                                                          & Application game                                                                                                         \\ \hline
$\Omega$                         & \begin{tabular}[c]{@{}l@{}}Set of all orderings of a (subset \\ of) transactions in $\X$\end{tabular}                            & $\Ga$                                                                                                         & Arbitrary game                                                                                                           \\ \hline
$\beta$                          & Block behaviour                                                                                                                  & \begin{tabular}[c]{@{}c@{}}$\sigma_{\bg/\ag/\nog/\cg},$\\ $\overline{\sigma}_{\bg/\ag/\nog/\cg}$\end{tabular} & \begin{tabular}[c]{@{}l@{}}strategy profile of an \\ application/network/cross-layer game\end{tabular}                   \\ \hline
$s,t$                            & Time indices                                                                                                                     & $\Sigma_{\bg/\ag/\nog/\cg}$                                                                                   & \begin{tabular}[c]{@{}l@{}}Set of strategy profiles of a blockchain/\\ application/network/cross-layer game\end{tabular} \\ \hline
$A_h$                            & Block at height $h$                                                                                                               & $\pi_{\bg/\ag/\nog}$                                                                                          & \begin{tabular}[c]{@{}l@{}}Outcome function of a blockchain/\\ application/network game\end{tabular}                     \\ \hline
$i,j$                            & \begin{tabular}[c]{@{}l@{}}Indices used to represent players,\\ $j$ is usually reserved for miners\end{tabular}                  & $(\Ga^p)_{p\in\P}$                                                                                            & Collection of parametrised games                                                                                         \\ \hline
$\omega$                         & Execution function                                                                                                               & $p,q,\P,\Q$                                                                                                   & \begin{tabular}[c]{@{}l@{}}Parameters and their respective\\ parameter spaces\end{tabular}                               \\ \hline
$M$                              & \begin{tabular}[c]{@{}l@{}}Set of blockchain game players \\ (miners/validators)\end{tabular}                                    & $\chi$                                                                                                        & Player function                                                                                                          \\ \hline
$\tx{}{}$                        & Transaction                                                                                                                      & $\No(Y)$                                                                                                      & Network game induced by $Y$                                                                                              \\ \hline
$f$                              & Fee function                                                                                                                     & $\F_{\Pi,\beta,\omega}$                                                                                       & \begin{tabular}[c]{@{}l@{}}Cross-layer game with respect \\ to $\beta$ and $\omega$ of $\Pi$\end{tabular}                \\ \hline
$\D(S)$                          & \begin{tabular}[c]{@{}l@{}}Set of probability distributions \\ over $S$\end{tabular}                                             & $\B(\G)$                                                                                                      & \begin{tabular}[c]{@{}l@{}}Collection of blockchain games \\ induced by outcome of $\G$\end{tabular}                     \\ \hline
$X$                              & Set of transactions                                                                                                              & $\Sigma_{\nog}(\sigma_{\ag}), \sigma_{\nog}(\sigma_{\ag})$                                                      & $\Sigma_{\nog}^{\pi_\ag^p(\sigma_\ag)}, \sigma_{\nog}^{\pi_\ag^p(\sigma_\ag)}$                                             \\ \hline
$\Y(X)$                          & \begin{tabular}[c]{@{}l@{}}Set of all transaction triples \\ of $X$\end{tabular}                                                 & $\pi_{\nog}^{(\sigma_\ag)}$                                                                                    & $\pi_{\bg}^{\pi_\ag^p(\sigma_\ag)}$                                                                                      \\ \hline
$\Po(S)$                         & Power set of $S$                                                                                                                 & $\sigma_{\nog}(\cdot)$                                                                                         & $(\sigma_{\nog}(\sigma_{\ag}'))_{\sigma_{\ag}'\in\Sigma_{\ag}^p}$                                                         \\ \hline
$Y$                              & Set of transaction triples                                                                                                       & $\sigma_{\bg}(\cdot,\cdot)$                                                                                   & $(\sigma_{\bg}(\sigma_{\nog}(\sigma_{\ag})))_{(\sigma_\ag,\sigma_{\nog}(\cdot))\in\Sigma_{\ag,\nog}^p}$                  \\ \hline
$b$                              & Blockchain ordering                                                                                                              & $P$                                                                                                           & Probability measure                                                                                                      \\ \hline
$\bmodel_M(Y)$                   & \begin{tabular}[c]{@{}l@{}}Set of possible blockchain orderings \\ with transactions from $Y$\end{tabular}                       & $u=(u_i)_{i\in N\cup M}$                                                                                      & Vector of utility functions for each $i$                                                                                 \\ \hline
$\B(Y)$                          & Blockchain game induced by $Y$                                                                                                   & $\eta$                                                                                                        & Collusion reduction transformation                                                                                       \\ \hline
$\lambdavec=(\lambda_j)_{j=0}^m$ & Hashrate distribution                                                                                                            & $\tilde{\Ga}$                                                                                                 & Collusion reduction of $\Ga$                                                                                             \\ \hline
$\B_{\Uup/\Gup/\Lup}(Y)$         & \begin{tabular}[c]{@{}l@{}}Unpredictable/Globally Predictable\\ /Locally Predictable Censor-Only \\ Blockchain Game\end{tabular} & $\gvec=(g_p)_{p\in\P}$                                                                                        & Collection of composition functions                                                                                      \\ \hline
$D$                              & Epoch length                                                                                                                     & $\circ_\gvec$                                                                                                 & $\gvec$-composition                                                                                                      \\ \hline
$e$                              & Epoch index                                                                                                                      & $\1_E$                                                                                                        & \begin{tabular}[c]{@{}l@{}}Indicator function; equals $1$ \\ if $E$ is true and $0$ otherwise\end{tabular}               \\ \hline
$\vvec^e=(v_t^e)_{t=0}^{D-1}$    & Miner assignment for epoch $e$                                                                                                   & $T$                                                                                                           & Timelock                                                                                                                 \\ \hline
$\vvec^e|_j$                     & \begin{tabular}[c]{@{}l@{}}Rounds for which $j$ is chosen \\ as leader by $\vvec^e$\end{tabular}                                 & $F_2=\kappa f_1$                                                                                              & \begin{tabular}[c]{@{}l@{}}Bribe budget, available to try to \\ censor a transaction paying fee $f_1$\end{tabular}       \\ \hline
$\mu_{j,t}$                      & \begin{tabular}[c]{@{}l@{}}Block proposal function for \\ miner $j$ at round $t$\end{tabular}                                    & $C$                                                                                                           & Event that $\tx{1}{}$ gets censored                                                                                      \\ \hline
$b_{j,t}$                        & \begin{tabular}[c]{@{}l@{}}Block proposal of miner $j$ \\ at round $t$\end{tabular}                                              & $N_0^T$                                                                                                       & \begin{tabular}[c]{@{}l@{}}Number of distinct validators in \\ rounds $0,\ldots,T$\end{tabular}                          \\ \hline
$\Pi$                            & Application protocol                                                                                                             & $\vee,\wedge$                                                                                                 & Maximum, minimum                                                                                                         \\ \hline
$N$                              & Application player set                                                                                                           & $g_j^t$                                                                                                       & \begin{tabular}[c]{@{}l@{}}Expected gain of miner $j$ \\ at round $t$\end{tabular}                                       \\ \hline
\end{tabular}
}
\end{strip}
}

    \section{Discussion on the worst case miner distribution}
\label{app:worst-case}

We argue informally that for a given $\lambda_0$ and $f_1<f_2$ the worst-case (i.e. largest) required timelock $T$ happens with two miners ($m=2$) and their hash rates are the same ($\lambda_1=\lambda_2$), i.e., when $\lambdavec=(\lambda_0,\frac{1-\lambda_0}{2},\frac{1-\lambda_0}{2})$. We plan to make this argument more formal in future work.

Let $\lambda_0,\ldots,\lambda_m >0$ such that $\sum_i \lambda_i =1$, where $\lambda_1 \le \lambda_2\le\cdots \le \lambda_m\leq\frac{1}{2}$ is the distribution of hashrates of $m>1$ miners, whereas $\lambda_0$ is the hashrate of the remaining myopic miners.
Recall that $t_j^*=T-\ceil{\rho_j}$ denotes the round from which on miner $j$ will censor $\tx{1}{}$, and note that $\rho_1\le \cdots\le \rho_m$ due to Corollary~\ref{cor:start-censor}. We will now motivate that compared to $m=2$ and $\lambda_1=\lambda_2$ we have that both  $m>3$ and $\lambda_1 <\lambda_2$ only decrease $\rho_m$. We do this by transforming the hashrate distribution step-by-step to include fewer and fewer miners. Denote by $\rho_{\max}$ the maximal $\rho_i$ in each step, as the number of miners $m$ will decrease in each step.

Our argument proceeds in two steps: first we transform it into a situation where the biggest miner has as least as much hashrate as the remaining non-myopic miners combined by merging the biggest miners (and therefore decreasing the number of miners by $1$). This step only increases $\rho_{\max}$. Secondly, if there exists one ``big'' miner, then merging smaller miners will only increase $\rho_{\max}$. This implies that $m=2$ yields the largest $\rho_{\max}$, which translates to the largest required timelock to have censoring probability zero.

First, we consider $m$ miners with hashrates $\boldsymbol{\lambda}$, which we will transform to $m-1$ miners by letting $\lambda'_{m-1}:= \lambda_{m-1} + \lambda_m$ and $\lambda_k' := \lambda_k$ for $k<m-1$. We  now argue that $\rho_{\max} \le \rho'_{\max}$, since the bigger hashrate $\lambda'_{m-1} > \lambda_m$ makes the biggest miner more resilient to the risk that comes with censoring. For the second step, we consider $m$ miners with hashrates~$\boldsymbol{\lambda}$, which we will transform to $m-1$ miners by letting $\lambda'_{m-2}:= \lambda_{m-2} + \lambda_{m-1}$ and $\lambda'_{m-1} := \lambda_m$ (assuming $\lambda'_{m-2}<\lambda'_{m-1}$; hence the first step) and $\lambda'_k := \lambda_k$ for $k<m-2$. By the above argument, $\rho_{\max} \le \rho'_{\max}$. 

Repeating these two steps to iteratively group either small miners or big miners eventually leads to just two miners. Lastly, analysing the case $m=2$ shows that $\rho_2$ as a function of $\lambda_2$ is increasing on $[0, \tfrac{1-\lambda_0}2)$ with an extremum at $\tfrac{1-\lambda_0}2$. Hence, $m=2$ with $\lambda_1=\lambda_2$ gives the largest $\rho_{\max}$.


    
    \section{Detailed Case Study: CRAB Payment Channels}
\label{app:timelock-examples}
Unlike the examples in Section~\ref{subsec:pc}, where the structure of the underlying PC was kept implicit, we now consider a specific construction in detail. In particular, Lightning-style channels are known to be vulnerable to bribing attacks when miners behave rationally~\cite{aumayr2024securing}. To address this, \cite{aumayr2024securing} introduced the CRAB channel, a design that mitigates such attacks. The core idea is to let the participants each put in an additional collateral $c$ into the channel, that will be rewarded to a miner in case a participant misbehaves by publishing an old channel state. The CRAB construction introduces a delay $T$ before which the party closing the channel cannot claim their balance, as illustrated in Figure~\ref{fig:crab}.
\begin{figure}[h]
    \centering
    \includegraphics[width=\linewidth]{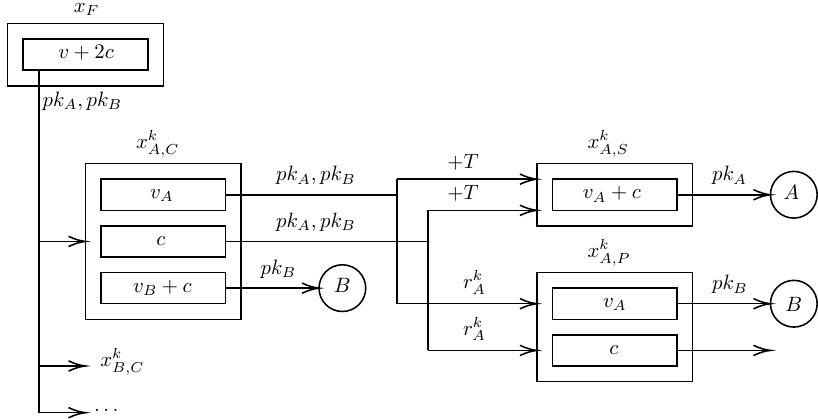}
    \caption{Transaction flow for a CRAB PC}
    \label{fig:crab}
\end{figure}
Just as in standard Lightning channels, the CRAB channel locks its capacity in a 2-of-2 multisignature output via a transaction $\tx{F}{}$. Each participant can unilaterally broadcast a commitment transaction that closes the channel and distributes the funds on-chain. In Figure~\ref{fig:crab}, we specified Alice's $k$-th commitment transaction $\tx{A,C}{k}$. These commitment transactions are regularly updated to reflect the latest balances, but earlier (outdated) versions remain valid and can still be broadcast. To discourage this, participants exchange revocation keys during updates, enabling the other party to claim the entire channel balance through a punishment transaction if an old commitment is posted. For example, when updating from state $k$ to state $k+1$, Alice will share $r_A^k$ with Bob and Bob will give $r_B^k$ to Alice. Because there is a timelock $T$, Alice cannot immediately post $\tx{A,S}{k}$ to claim her funds, giving Bob the opportunity to punish her by posting $\tx{A,P}{k}$\footnote{In the actual CRAB construction, $\tx{A,P}{k}$ is not specified, and instead replaced by two transactions; one allowing Bob to claim $v_A$ and one allowing anyone to claim $c$. Specifying $\tx{A,P}{k}$ is just for simplification, and does not alter the security analysis as including only the transaction that would spend $c$ renders $\tx{A,S}{k}$ already unspendable.} in case state $k$ would be an old state. The second output $c$ in $\tx{A,P}{k}$ can be claimed by anyone, and will thus be claimed by the miner who includes $\tx{A,P}{k}$. 

In standard Lightning, rational miners may be bribed to include an outdated commitment transaction while censoring the corresponding punishment transaction—resulting in fund loss for the honest party. CRAB channels mitigate this by requiring both participants to lock up collateral, which is forfeited to the miner if a party attempts to close the channel dishonestly using an old commitment. Using our framework, we analyse under which conditions such misbehaviour becomes unprofitable with respect to U-COBG.

Consider a CRAB channel between Alice and Bob. Alice wishes to close the channel and can choose to broadcast either her most recent commitment transaction $\tx{A,C}{l}$, reflecting balances $v_A^l$ and $v_B^l$, or an outdated one $\tx{A,C}{o}$, reflecting $v_A^o$ and $v_B^o$. She may also choose whether to attach a corresponding spending transaction $\tx{A,S}{l}$ or $\tx{A,S}{o}$. If Alice posts $\tx{A,C}{o}$, Bob can retaliate by broadcasting a punishment transaction $\tx{A,P}{o}$, which gives him Alice’s balance and the miner her collateral. Alice can try to prevent this by raising the fee on $\tx{A,S}{o}$, effectively bribing miners to ignore the punishment. Bob may counter-bribe by increasing the fee on his punishment transaction. This fee race continues until the timelock expires, allowing Alice to settle using $\tx{A,S}{o}$ if the punishment was successfully censored.

We model this interaction as an application game $\mathcal{K}$ with player set $N = \qty{A, B}$, and $(X, \Sigma_\ag, \pi_\ag)$ defined via a recursive extensive-form game. The game unfolds as a sequence $(K_t)_{t = -1}^T$, illustrated in Figures~\ref{fig:crab-tree1} and \ref{fig:crab-tree2}.
\begin{figure}[h]
    \centering
    \includegraphics[width=0.6\linewidth]{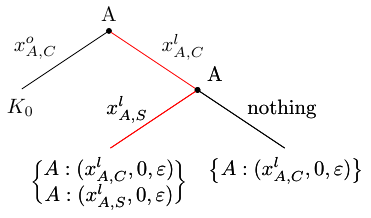}
    \caption{Tree for $K_{-1}$.}
    \label{fig:crab-tree1}
\end{figure}

\begin{figure}[h]
    \centering
    \includegraphics[width=\linewidth]{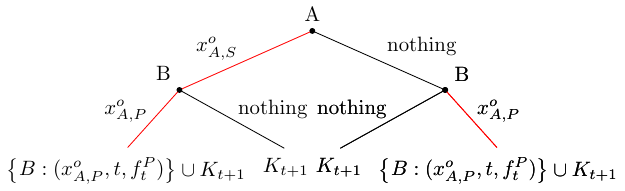}
    \caption{Tree for $K_t$ for $0\leq t\leq T$. If $t=T$, replace $K_{t+1}$ by $\varnothing$.}
    \label{fig:crab-tree2}
\end{figure}

We can prove a sufficient condition for Alice not posting old commitment transactions in the parametrised cross-layer game $(\mathcal{K}^T,\bgame_{0}(\mathcal{K}^T),\omega)$, by recognising once again a pre-BitVM timelock bribing risk. In this cross-layer game, $\omega$ can be specified via a function $\tilde{\omega}$, such that $\omega_i(b)=\tilde{\omega}_i(b)-\sum_{(x,t,f)\in b:x\in\chi(i)}\sum_{j\in M}f_j(b)$, for all $i\in N$ and each $b\in\bmodel_M(Y)$ and $\omega_j(b)=\tilde{\omega}_j(b)+\sum_{(x,t,f)\in b}f_j(b)$, for all $j\in M$ and $b\in\bmodel_M(Y)$. Specifically, $\tilde{\omega}$ is given by:
\begin{align*}
    \tilde{\omega}_A(b)&=
    \begin{cases}
        v_A^l+c, & b\supseteq\{\tx{A,C}{l},\tx{A,S}{l}\}, \\
        v_A^o+c, & b\supseteq\{\tx{A,C}{o},\tx{A,S}{o}\}, \\
        0, & b\supseteq\{\tx{A,C}{o},\tx{A,P}{o}\}, 
    \end{cases} \\
    \tilde{\omega}_B(b)&=
    \begin{cases}
        v_B^l+c, & b\supseteq\{\tx{A,C}{l}\}, \\
        v_B^o+c, & b\supseteq\{\tx{A,C}{o}\}, \\
        v+c, & b\supseteq\{\tx{A,C}{o},\tx{A,P}{o}\}, 
    \end{cases} \\
    \tilde{\omega}_j(b)&=
    \begin{cases}
        c, & b\supseteq\{\tx{A,C}{o},\tx{A,P}{o}\} \\
        0, & \text{otherwise},
    \end{cases}
\end{align*}
where $j$ is the miner who included $\tx{A,P}{o}$. With the notation $b\supseteq S$ we mean every blockchain ordering that \emph{prioritises} the transactions in $S$ over any possible transactions that might spend the same funds as the transactions in $S$, by ordering the former transactions before the latter.

\begin{theorem}[\ref{subsec:thm:crab}]
\label{thm:crab}
    In $(\mathcal{K}^T,\bgame_{0}(\mathcal{K}^T),\omega)$, Alice will not post an old commitment transaction if $T>\ceil{\rho_m(c,v;\lambdavec)}$, where $\lambdavec$ is the hashrate distribution assumed in $\bgame_{0,\mathcal{K}^T}$.
\end{theorem}
\begin{remark}
\label{rem:crab}
    If we assume the worst-case hashrate distribution $\lambdavec=(0,\frac{1}{2},\frac{1}{2})$, we have $\rho_1=\rho_2$ equal to $0$ if $\frac{c}{v}>\frac{1}{2}$ and equal to $\infty$ otherwise, so we retrieve the same condition as \cite{aumayr2024securing} that $c>\frac{v}{2}$ in order for Alice not to post an old commitment transaction. We stress once again that we can easily keep the application game as is and study the CRAB channel under different blockchain assumptions without changing the proof.
\end{remark}

\section{Detailed Case Study: Sandwich Attack}
\label{app:mev}
To recap, we consider an abstract model of a DeFi smart contract where a user $U$ makes a trade by placing a buy limit order, i.e., a transaction $\tx{U}{}$, that buys some asset for a price up to $l$. The current price is $l-p$, with $p\geq0$, and so if $\tx{U}{}$ were included as is, $U$ would have utility $p$. However, the miners are also able to interact with the smart contract and are thus part of the application game. Upon seeing $\tx{U}{}$, each miner has the capability to front- and backrun, or \emph{sandwich}, $\tx{U}{}$. In particular, each miner $j\in M$ could add a transaction $\tx{F,j}{}$ before $\tx{U}{}$ buying the asset at price $l-p$, and a transaction $\tx{B,j}{}$ after $\tx{U}{}$, selling the asset again at price $l$. In doing so, the miner captures $p$ and $U$ is left with utility $0$, as the transaction now buys at price $l$ instead of $l-p$. 

We model this by playing the application game $\G=(\qty{U}\cup M,\{\tx{U}{}\}\cup\{\tx{F,j}{},\tx{B,j}{}\}_{j\in M},\Sigma_\ag,\pi_\ag)$. The set $M$ contains the miners, who are all part of the application game and can thus monitor the intents of other players. The strategy profile set $\Sigma_\ag$ reflects this by corresponding to a simultaneous game in which each miner gets to decide whether or not construct sandwiching transactions, upon seeing $\tx{U}{}$. We will assume that $(\tx{U}{},0,f)\in\pi_\ag(\sigma_\ag)$ for each $\sigma_\ag\in\Sigma_\ag$, with $f>0$ some fixed fee. Afterwards, we play for a given $\pi_\ag(\sigma_\ag)$ a network game where 
$U$ gets to decide to which miners to send $(\tx{U}{},0,f)$, i.e., the network game $\mathcal{N}=(\qty{U},M,\pi_\ag(\sigma_\ag),\Po(M),\pi_{\nog})$, where $\sigma_{\nog}\in\Po(M)$ is the strategy profile containing the set of miners that $U$ chose to share $\tx{U}{}$ with, and $\pi_{\nog}(\sigma_{\nog})$ is the tuple $(\pi_{\nog,j}(\sigma_{\nog}))_{j\in M}$ where for each $j\in M$, $\pi_{\nog,j}(\sigma_{\nog})$ is defined as $\pi_\ag(\sigma_\ag)\cap\{(\tx{F,j}{},0,0),(\tx{B,j}{},0,0)\}$ if $j\in\sigma_{\nog}$ and $\pi_\ag(\sigma_\ag)\cap\{(\tx{F,j}{},0,0),(\tx{B,j}{},0,0)\}\setminus\{(\tx{U}{},0,f)\}$ otherwise. Remark that although the miners might already know $U$ wants to make the trade $\tx{U}{}$, they would only be willing to sandwich $\tx{U}{}$ if they have actually received $\tx{U}{}$ as input to the blockchain game.

For a given $\pi_{\nog}(\sigma_{\nog})$, the subsequent blockchain game is given by $(M,\pi_{\nog}(\sigma_{\nog}),\Sigma_{\bg},\pi_{\bg})$\footnote{Notice the abuse of notation, $\pi_{\nog}(\sigma_{\nog})$ is a tuple of sets of transaction triples.}. We specify $\Sigma_{\bg}=\qty{H,D}^\qty|M|$, where each miner decides whether to ($H$) just include $\tx{U}{}$, or to ($D$) sandwich $\tx{U}{}$, if they were to receive the transaction $\tx{U}{}$. One miner $j\in M$ then gets selected at random (proportional to its hashrate) to propose the blockchain ordering, which will either be $b=b_{H,j}$ (containing only $\tx{U}{}$) or $b=b_{D,j}$ (containing $\tx{F,j}{},\tx{U}{},\tx{B,j}{}$). This implicitly defines the ordering function $\pi_{\bg}$. We moreover specify, for each possible ordering $b$, the execution function $\omega$ for $U$ as $\omega_U(b)=p$ if $b=b_{H,j}$ for some $j\in M$ and $\omega_U(b)=0$ if $b=b_{D,j}$ for some $j\in M$, and for each $j\in M$ as $\omega_j(b)=f\1_\qty{b=b_{H,j}}+(f+p)\1_\qty{b=b_{D,j}}$.

From the above formulation, it is clear that without any additional measures, the miners will always choose strategy $D$. Given that, $U$ will always end up with zero utility and does not have any preference with regards to which miner to share $(\tx{U}{},0,f)$ with. A simple countermeasure to this inevitable MEV scenario would be to introduce a trusted miner, holding a proportion $\lambda_H$ of the hashrate. We can model this miner as a player with only the strategy $H$ (or equivalently assign a utility of $-\infty$ to any strategy profile where this miner has strategy $D$). Consequently, it is clear that $U$ will only share $(\tx{U}{},0,f)$ with the trusted miner. 

This section examines how limited transaction propagation at the network layer impacts miner incentives and strategic behavior, extending beyond a simple model where all miners see all transactions. Focusing on a DeFi scenario where a user places a buy limit order vulnerable to miner “sandwich” attacks (front- and back-running to capture user profits), the model shows that miners who receive the user’s transaction can exploit it, leaving the user with zero utility. Because miners will always choose to sandwich when possible, the user gains nothing by sharing the transaction broadly. A key insight is that introducing a trusted miner who commits not to sandwich can restore some utility to the user, as the user will only share the transaction with this miner, highlighting how selective transaction propagation and trusted parties can shape MEV outcomes and strategic interactions.
    
    \section{Further Types of Composition}
\label{app:examples}
\subsubsection*{\textbf{Composition Across Different Applications}}
We begin by illustrating how our framework supports reasoning about the interaction of distinct application-layer protocols that execute concurrently on a shared blockchain infrastructure. Such interactions are increasingly prevalent in decentralised finance (DeFi), where composability enables sophisticated behaviour but also introduces subtle incentive misalignments. We present two representative examples, Cross-DEX arbitrage and Oracle-to-DEX feedback loops, that highlight how our framework can formally capture cross-application dependencies and enable compositional security analysis that would be difficult to achieve in isolation.

\textit{Cross-DEX Arbitrage.}
Arbitrage strategies across decentralised exchanges (DEXs) may span multiple smart contracts executing in parallel on the same blockchain. While individual DEXs might appear secure in isolation, the compositional perspective reveals how IC behaviour at one layer (e.g., submitting a buy order) can enable profit extraction in another (e.g., selling on a second DEX before the price adjusts). Our framework captures these interactions as separate games composed over a shared blockchain layer, making it possible to analyse whether such arbitrage behaviours are expected, profitable, or potentially destabilising under different network and consensus assumptions.

\textit{Oracle-DEX Feedback Loops.}
Oracles supply real-world data to smart contracts, and DEXs may use this data to trigger conditional logic such as liquidations or price-based execution. This setup creates a feedback loop: oracle updates influence DEX behaviour, which in turn can incentivise manipulation of the oracle. Our framework captures this setting by representing the oracle and DEX as two parametrised application games, composed through a dependency on a shared data stream. By instantiating different network games or blockchain models (e.g., censorship-tolerant vs. adversarial ordering), we can formally analyse the conditions under which oracle manipulation becomes profitable and explore mitigation strategies such as delay windows or cross-validation mechanisms.

\subsubsection*{\textbf{Cross-Blockchain Games: Atomic Swaps}}
Atomic swaps are a class of protocols that enable trustless exchange of assets between two different blockchains—such as Bitcoin and Ethereum—without intermediaries. Modelling such protocols requires reasoning about two independent blockchain environments, each with its own set of participants, timing assumptions, and strategic behaviour. Our framework naturally extends to this setting by treating each blockchain as a separate but composable component.

To model an atomic swap, we instantiate two distinct blockchain games: $\bgame^{(1)}$ and $\bgame^{(2)}$, corresponding to the execution environments of the two blockchains. Each blockchain game has its own player set, transaction structure, and execution function. Likewise, the protocols executed on these blockchains—e.g., Hashed Timelock Contracts (HTLCs) on both chains—are modelled as separate parametrised application games $\G_1=(\G_1^p)_{p\in\P}$ and $\G_1=(\G_2^q)_{q\in\Q}$, possibly sharing overlapping players. For simplicity, assume we have in both games players sets containing Alice and Bob.

The interaction between the two chains is captured via a shared parameter, typically a secret $s$ revealed through the execution of a transaction on one blockchain. For example, Alice posts a transaction on $\bgame^{(1)}$ that discloses $s$, enabling Bob to claim funds on $\bgame^{(2)}$ before a timelock expires. We express this dependency by considering the composition $\G_2\circ_\gvec \G_1$, where similarly to the collection $\gvec^{dep}$ in Section~\ref{subsec:pc}, we define $\gvec$ as $(g_p)_{p\in\P}$ where for all $p\in\P$, $g_p(\sigma_{\ag,1})$ is a function of the time $\tau_A(\sigma_{\ag,1})$ at which Alice discloses $s$ in case of strategy profile $\sigma_{\ag,1}\in\Sigma_{\ag,1}^p$. 

This yields two cross-layer games $(\G_1^p, \No(\G_1^p), \allowbreak \bgame^{(1)}(\No),\omega_1)$ and $(\G_2^p, \No(\G_2^p), \bgame^{(2)}(\No), \omega_2)$, whose compositional analysis allows us to study incentive compatibility in the joint system. Specifically, our framework allows for evaluating when rational participants would follow the intended swap protocol or deviate—e.g., by claiming coins on one chain without enabling the counterparty to do so on the other.

A key advantage of our approach is that each blockchain can be modelled with different network assumptions, execution rules, or adversarial capabilities. As such, one can study swap security under heterogeneous conditions—such as delayed propagation, partial censorship, or differing hashrate distributions—while preserving a modular and rigorous incentive analysis. This illustrates how our framework extends beyond a single chain to capture the growing class of cross-chain protocols in a principled and compositional manner.


\subsubsection*{\textbf{Composing Complex Application-Network-Consensus Layer: MEV Auctions and PBS}}
Proposer-Builder Separation (PBS) is a design principle introduced to mitigate MEV centralisation in proof-of-stake systems like Ethereum. In PBS, specialised actors called \emph{builders} aggregate user transactions into blocks and participate in out-of-protocol \emph{block auctions} to sell these blocks to validators or proposers, who ultimately publish them on-chain. This decouples the roles of transaction inclusion and block finalisation, creating a complex multi-agent, cross-layer incentive environment.

Our framework naturally captures the layered structure of PBS. First, we model the \emph{searchers}—who scan mempools and applications for profitable transaction orderings—as players in one or more application games. Each application game produces a set of transaction triples based on strategic behaviour, including MEV-seeking bundles. These bundles are passed into a \textit{network game} $\mathcal{N}^{\text{PBS}}$, representing the builder auction layer. In this game, multiple builders (players) receive bundles from searchers and compete to construct block proposals. Each builder's strategy is to select and order a subset of received transactions into a candidate block, compute a bid (i.e., the payment they are willing to offer to the proposer), and submit this to the auction.

The proposer is modelled either as a designated player in the blockchain game or as part of the network  game logic. It selects the highest-bidding builder and forwards its block to the consensus layer, inducing the blockchain ordering $b$. The corresponding \emph{blockchain game} $\bgame$ thus reflects not direct transaction submission, but the output of the PBS auction, mediated by proposer selection rules and builder behaviour.

The \textit{execution function} $\omega$ maps the resulting blockchain ordering to utilities for all participants: searchers receive utility if their bundles are included profitably, builders earn the difference between the bid and their internal profit margin, and the proposer earns the winning bid.

This compositional structure allows us to: (i) model \emph{timing asymmetries} (e.g., builders with faster access to bundles), (ii) incorporate \emph{collusion or exclusive order flow} via $\mathcal{N}^{\text{PBS}}$, (iii) reason about \emph{incentive compatibility} of proposer or builder strategies, and (iv) study the effect of \emph{network propagation models} (e.g., which searchers reach which builders) on MEV extraction.

Crucially, since application games are parametrised, our framework supports scenarios where the utility of one application-layer actor (e.g., a searcher exploiting a DEX arbitrage) depends on the success of another (e.g., a liquidator in a lending protocol), highlighting cross-application interactions common in real-world MEV.

By explicitly modelling the PBS architecture as a compositional game, we enable principled analysis of incentive alignment between searchers, builders, and proposers. This facilitates exploration of mechanism design questions such as: what auction formats discourage censorship? How should rebates or fees be structured to incentivise honest builder behaviour? And how robust are MEV mitigation techniques under strategic deviation? Our framework provides the tools to formally investigate such questions within a layered game-theoretic abstraction.

    \section{Proofs}
\label{app:proofs}

\subsection{Proof of Theorem~\ref{thm:constant-composition}}
\label{subsec:thm:constant-composition}
For ease of notation, we will for this proof write $\sigma$ and $\Sigma$ instead of $\sigma_\ag$ and $\Sigma_\ag$ for strategies and strategy profiles in the application game. Let us first define the concept of a projection on a player subset, which will be useful in a bit.
\begin{definition}[Projection on a player subset]
\label{def:eta-projection}
    For a player set $N$, and a subset $S\subseteq N$, we say that the transformation $\eta_S:S\to S$ is a \emph{projection on S} of the transformation $\eta:N\to N$, if for every $i\in S$ we have 
    \begin{equation*}
        \eta_S(i)=
        \begin{cases}
            \eta(i),&\eta(i)\in S,\\
            \zeta(\eta(i)),&\eta(i)\notin S,
        \end{cases}
    \end{equation*}
    for some injective map $\zeta:\eta(S)\setminus S\rightarrow S\setminus\eta(S)$, defined whenever $\eta(S)\setminus S\neq\varnothing$.
\end{definition}

It suffices to show that for each value of the parameter $p\in\P$, $(((\G_2\circ_\gvec\G_1)^p,\beta,\omega),\overline{\sigma}^p)$ is IC w.r.t $\beta$ (implying player set $M$), where $\gvec$ is a collection $(g_p)_{p\in\P}$ of constant functions. We proceed by contradiction. Fix some $p\in\P$, let $\eta:N\cup M\to N\cup M$ (with $N=N_1\cup N_2$) be some transformation, and denote by $\tilde{\G}$ the $\eta$-CR of $\G:=(\G_2\circ_\gvec\G_1)^p$. We will derive a contradiction by assuming that the strategy profile $\overline{\sigma}^p:=(\overline{\sigma}_{1}^p,(\overline{\sigma}_{2}^{g_p(\sigma_{1})})_{\sigma_{1}\in\Sigma_{1}^p})$, which we will refer to as $\overline{\sigma}=(\overline{\sigma}_{1},\overline{\sigma}_{2})$ as $p\in\P$ is fixed and $g_p(\overline{\sigma}_{1})=q_p$ for some constant $q_p\in\Q$, is not a Nash equilibrium for $(\tilde{\G},\beta,\omega)$. That is, in $(\tilde{\G},\beta,\omega)$ there exists a deviation $\sigma=(\sigma_{1},\sigma_{2})$ by a player $k\in N\cup M$ such that $u_k(\sigma)>u_k(\overline{\sigma})$.

Let us denote by $\eta_1$ and $\eta_2$ the projections of $\eta$ on $N_1\cup M$ and $N_2\cup M$, as defined in Definition~\ref{def:eta-projection}. These projections define CRs $\tilde{\G}_1$ of $\G_1$ and $\tilde{\G}_2$ of $\G_2$, respectively. It should be clear by additivity of the execution function, and by the fact that the utility of a collusion of players is the sum of the utilities of the individual players in that collusion, that $u_k(\sigma)=u_{1,k_1}(\sigma_{1})+u_{2,k_2}(\sigma_{2})$, and $u_k(\overline{\sigma})=u_{1,k_1}(\overline{\sigma}_{1})+u_{2,k_2}(\overline{\sigma}_{2})$, where $k_1=\eta_1(\eta^{-1}(k)\cap (N_1\cup M))$ and $k_2=\eta_2(\eta^{-1}(k)\cap (N_2\cup M))$. But then, by assumption, $u_{1,k_1}(\sigma_{1})+u_{2,k_2}(\sigma_{2})>u_{1,k_1}(\overline{\sigma}_{1})+u_{2,k_2}(\overline{\sigma}_{2})$, which means that $u_{1,k_1}(\sigma_{1})>u_{1,k_1}(\overline{\sigma}_{1})$ or $u_{2,k_2}(\sigma_{2})>u_{2,k_2}(\overline{\sigma}_{2})$. But then, $\sigma_{1}$ is a profitable deviation by $k_1$ from $\overline{\sigma}_{1}$ in $(\tilde{\G}_1,\beta,\omega)$, or $\sigma_{2}$ is a profitable deviation by $k_2$ from $\overline{\sigma}_{2}$ in $(\tilde{\G}_2,\beta,\omega)$, contradicting $\Pi_1$ and $\Pi_2$ being both IC w.r.t. $\beta$.

\subsection{Proof of Theorem~\ref{thm:g-cobg}}
\label{subsec:thm:g-cobg}
First of all, one should realise that for a given leader schedule $(v_t)_{t=0}^T$\footnote{We leave out the epoch superscript as we are always in epoch 0.}, censoring can only occur if each validator in the leader schedule is sufficiently bribed. That is, a validator $j$ who is selected to propose a block before or at round $T$, will only censor $\tx{1}{}$ if $j$ is guaranteed a reward larger than $f_1$. The fee function $f_2$ must therefore give more than $f_1$ to each \emph{distinct} validator in the leader schedule. Everyone can see whether $f_2$ actually does that, and if not, already $v_0$ will include $\tx{1}{}$ as it knows that the censoring will not succeed. Therefore, for a given leader schedule $(v_t)_{t=0}^T$ with $n$ distinct validators, a bribe budget $F_2$ enables censoring only if $F_2>nf_1$, via the fee function $f_2$ described earlier. 

If the fee function with bribe budget $F_2:=\kappa f_1$ is now fixed before the leader schedule is known, censoring will happen only if $\kappa$ is sufficiently large to cover the number of distinct miners, which happens with a certain probability. We calculate this probability below, denoting by $Z_t$ the event that $v_t=0$.
\begin{align*}
    P(C)&=P\qty(C\bigg|\bigcap_{t=0}^{T}Z_t^c)P\qty(\bigcap_{t=0}^{T}Z_t^c)\\&\quad+P\qty(C\bigg|Z_T\cap\bigcap_{t=0}^{T-1}Z_t^c)P\qty(Z_T\cap\bigcap_{t=0}^{T-1}Z_t^c)\\&\quad+P\qty(C\bigg|\bigcup_{t=0}^{T-1}Z_t)P\qty(\bigcup_{t=0}^{T-1}Z_t)\\&=P(N_0^T<\kappa)(1-\lambda_0)^{T+1}\\&\quad+P(N_0^{T-1}<\kappa-1)\lambda_0(1-\lambda_0)^{T}\\&\quad+0\cdot P\qty(\bigcup_{t=0}^{T}Z_t)\\&=(1-\lambda_0)^{T}\qty[(1-\lambda_0)p_{T}(\kappa)+\lambda_0p_{T-1}(\kappa-1)],
\end{align*}
where in the second equality we used for the first term that if no single block in the rounds $0,\ldots,T$ were given to the portion $\lambda_0$, censoring would only occur if throughout rounds $0,\ldots,T$ there were less than $\kappa$ distinct validators, for the second term that if only $v_T$ is zero, at least $f_1$ should have been spent there in order for $\tx{2}{}$ to be preferred over $\tx{1}{}$, meaning that the remaining blocks $0,\ldots,T-1$ can only have been mined by less than $\kappa-1$ distinct validators, and for the third term that if any block in round $0,\ldots,T$ was mined by $\lambda_0$, $\tx{1}{}$ would have been included at that round and censoring would have been impossible.

\subsection{Proof of Theorem~\ref{thm:u-cobg}}
\label{subsec:thm:u-cobg}
(i) In the most optimistic case for the bribing party, all miners $1,\ldots,m$ will censor $\tx{1}{}$ from round $0$ up to round $T$. If we assume that in round $T$ $\tx{2}{}$ is included with probability one, censoring can only occur if every block in the rounds $t=0,\ldots,T-1$ is actually mined by a miner $j=1,\ldots,m$. This happens with probability $(1-\lambda_0)^T$. The actual censoring probability can thus only be smaller than or equal to $(1-\lambda_0)^T$.

(ii) According to \eqref{eq:u-cobg}, a miner $j=1,\ldots,m$ will only censor in round $0$ if its expected gain $g_j^1(j)$ is larger than $f_1$. Once again, assuming the best-case scenario for the bribing party in which all the miners $j=1,\ldots,m$ are censoring from round $1$ until round $T$, $\tx{2}{}$ will be included with probability at most $(1-\lambda_0)^{T-1}$, in which case miner $j$ will gain a reward of at most $F_2$. Hence, $g_j^1(j)\leq (1-\lambda_0)^{T-1}F_2$. If $\kappa<(1-\lambda_0)^{-T+1}$, we thus have $g_j^1(j)\leq(1-\lambda_0)^{T-1}F_2<\kappa^{-1}F_2=f_1$, which means that $j$ will not censor in round $0$. As this argument holds for every miner $j=1,\ldots,m$, no single miner will censor $\tx{1}{}$ in round $0$, guaranteeing $\tx{1}{}$ will be included.

\subsection{Proof of Theorem~\ref{thm:tau}}
\label{subsec:thm:tau}
For ease of notation, we will for this proof write $\sigma$ instead of $\sigma_{\bg}$ for strategies and strategy profiles in U-COBG. We will work our way backwards, starting at round $T$. For each round $t=0,\ldots,T$, and for each miner $j=1,\ldots,m$, we will compute the expected gain $g_j^t\qty(\sigma_j^t;(\sigma_j^s)_{s=t+1}^{T})$ (essentially acting as if $\tx{1}{}$ was censored until round $t$), 
where $\sigma_j^s\in\qty{1,2}$ denotes the decision taken by miner $j$ at time $s$, where $\sigma_j^s=1$ encodes the miner trying to include $\tx{1}{}$ in round $s$, and $\sigma_j^s=2$ encodes the miner censoring $\tx{1}{}$ in round $s$, or including $\tx{2}{}$ if possible. 

For miner $j=1,\ldots,m$, we have in round $T$ an expected gain $g_j^{T}\qty(\sigma_j^{T})$ of $\lambda_j f_1$ if $\sigma_j^{T}=1$ and of $\lambda_j f_2$ if $\sigma_j^{T}=2$. By assumption, $f_1<f_2$ and so in round $T$ all miners will include $\tx{2}{}$. Keep in mind that in U-COBG, we also consider a portion of the hashrate $\lambda_0$ that will by construction always include $\tx{1}{}$ in rounds before $T$, and will switch to including $\tx{2}{}$ only in round $T$ as it is the more rewarding transaction to include.

For rounds $t<T$, the expected gain depends on what players expect to gain in subsequent rounds. For example, at a round $t<T$, a miner $j$ can decide whether to include $\tx{1}{}$. Out of the remaining miners, there will be a portion of hashrate $\lambda_{I,j}^t$ that will include $\tx{1}{}$ at round $t$, and a portion of hashrate $\lambda_{C,j}^t$ that will censor $\tx{1}{}$ at round $t$. Hence, with probability $\lambda_{I,j}^t$, in round $t$ another miner will mine a block including  $\tx{1}{}$, leaving $j$ with a zero gain. With probability $\lambda_{C,j}^t$, another miner will mine a block in round $t$ excluding $\tx{1}{}$, leading to miner $j$ having to decide in round $t+1$ again whether or not to include $\tx{1}{}$. If we assume that $j$'s decision in round $t+1$ leads to an expected gain of $g$, $j$ will have with probability $\lambda_{C,j}^t$ a gain of $g$. This reasoning leads us to recursive expressions for the expected gain $g_j^{t}\qty(\sigma_j^t;(\sigma_j^s)_{s=t+1}^{T})$ in round $t=0,\ldots,T-1$; 
\begin{footnotesize}
\begin{align}
\label{eq:recursion}
    &g_j^{t}\qty(\sigma_j^t;(\sigma_j^s)_{s=t+1}^{T})=\nonumber\\&
    \begin{cases}
        \lambda_j f_1+\lambda_{C,j}^t g_j^{t+1}\qty(\sigma_j^{t+1};(\sigma_j^s)_{s=t+2}^{T}),& \sigma_j^t=1,\\
        \lambda_j g_j^{t+1}\qty(\sigma_j^{t+1};(\sigma_j^s)_{s=t+2}^{T})+\lambda_{C,j}^t g_j^{t+1}\qty(\sigma_j^{t+1};(\sigma_j^s)_{s=t+2}^{T}),& \sigma_j^t=2.
    \end{cases}
\end{align}
\end{footnotesize}
Hence, if $f_1<g_j^{t+1}\qty(\sigma_j^{t+1};(\sigma_j^s)_{s=t+2}^{T})$, $j$ will choose to censor $\tx{1}{}$ instead of including it.

For each miner $j=1,\ldots,m$, we are interested in the strategy $(\sigma_j^s)_{s=0}^{T}$ that maximises the expected gain $g_j^0$. First of all, we claim that this optimal strategy is always of the the form $(1,\ldots,1,2,\ldots,2)$. That is, a miner will always include $\tx{1}{}$ up to some point, after which this miner switches to censoring $\tx{1}{}$, never trying to include it again. 
\begin{lemma}
\label{lem:first1then2}
    For a fixed miner $j=1,\ldots,m$, there exists a round $t_j^*\in\qty{0,\ldots,T}$ such that the strategy $(\overline{\sigma}_j^t)_{t=0}^{T}$ that maximises the expected gain $g_j^0$ is given by
    \begin{equation}
    \label{eq:optimal-strat}
        \overline{\sigma}_j^t=
        \begin{cases}
            1,& t< t_j^*,\\
            2,& t\geq t_j^*.            
        \end{cases}
    \end{equation}    
\end{lemma}
\begin{proof}
    Clearly, since $f_1<f_2$, $\overline{\sigma}_j^{T}=2$. Let us now assume that for some $t\in\qty{1,\ldots,T}$, and given a sequence $(\overline{\sigma}_j^s)_{s=t+1}^{T}$, we have $\overline{\sigma}_j^{t}=1$. By Equation~\eqref{eq:recursion}, this implies that $f_1\geq g_j^{t+1}\qty((\overline{\sigma}_j^s)_{s=t+1}^{T})$. Again by Equation~\eqref{eq:recursion}, we find that $g_j^{t}\qty(\overline{\sigma}_j^t;(\overline{\sigma}_j^s)_{s=t+1}^{T})=(\lambda_j+\lambda_{C,j}^t)g_j^{t+1}\qty(\overline{\sigma}_j^{t+1};(\overline{\sigma}_j^s)_{s=t+2}^{T})\leq(\lambda_j+\lambda_{C,j}^t)f_1\leq f_1.$
    But then, using Equation~\eqref{eq:recursion} one last time for round $t-1$, we find $g_j^{t-1}\qty(1;(\overline{\sigma}_j^s)_{s=t}^{T})\geq g_j^{t-1}\qty(2;(\overline{\sigma}_j^s)_{s=t}^{T})$. Hence, $\overline{\sigma}_j^{t-1}=1$. By an induction argument, this implies that if we have $\overline{\sigma}_j^{t}=1$ at some round $t\in\qty{0,\ldots,T-1}$, we have $\overline{\sigma}_j^{s}=1$ for all $s<t$ as well. By contraposition and a similar induction argument, if we have $\overline{\sigma}_j^{t}=2$ at some round $t\in\qty{0,\ldots,T-1}$, we must also have $\overline{\sigma}_j^{s}=2$ for all $s>t$. Consequently, the optimal strategy must be of the form specified in \eqref{eq:optimal-strat}.
\end{proof}
By Lemma~\ref{lem:first1then2}, determining the optimal strategy for each miner amounts to finding the round $t_j^*$ at which miner $j$ switches from including $\tx{1}{}$ to censoring $\tx{1}{}$, for each $j=1,\ldots,m$. Recalling that $g_j^T(2)=\lambda_jf_2$ and combining \eqref{eq:recursion} and \eqref{eq:optimal-strat}, we know that for any miner $j=1,\ldots,m$, the expected gain for $t\geq t_j^*$ is given by
\begin{equation}
\label{eq:censoring-gain}
    g_j^{t}\qty((\overline{\sigma}_j^s)_{s=t}^{T})=\qty(\prod_{s=t}^{T-1}\lambda_C^s)\lambda_jf_2,
\end{equation}
where for $t\geq t_j^*$, we defined $\lambda_C^t=\lambda_j+\lambda_{C,j}^t$ as the total portion of hashrate that is censoring $\tx{1}{}$ in round $t$. Now remark that if we have $\overline{\sigma}_j^t=2$, we have $g_j^{t+1}\qty((\overline{\sigma}_j^s)_{s=t+1}^{T})>f_1$. Consequently, for any miner $k\in\qty{j+1,\ldots,m}$, which has by Definition~\ref{def:nfnwmg} a hashrate $\lambda_k\geq\lambda_j$, \eqref{eq:censoring-gain} implies that $g_k^{t+1}\qty((\overline{\sigma}_k^s)_{s=t+1}^{T})=\qty(\prod_{s=t+1}^{T-1}\lambda_C^s)\lambda_kf_2\geq \qty(\prod_{s=t+1}^{T-1}\lambda_C^s)\lambda_jf_2=g_k^{t+1}\qty((\overline{\sigma}_j^s)_{s=t+1}^{T})>f_1$, where the first equality holds whenever $t+1\geq t_k^*$. That is, $\overline{\sigma}_k^t=2$. By induction, one can easily find that larger miners will start censoring $\tx{1}{}$ before smaller miners.
\begin{corollary}
\label{cor:start-censor}
    Assume that $\lambda_j\leq\lambda_k$ for some $j,k\in\qty{1,\ldots,m}$. Then $\overline{\sigma}_j^t=2$ implies that $\overline{\sigma}_k^t=2$.
\end{corollary}

Starting at $t=T$ and working our way backwards, we will have an ever-shrinking set of censoring miners, where going further back in time will result in the smallest miner that is still censoring switching to including $\tx{1}{}$. To compute the times $t_j^*$ for all $j=1,\ldots,m$, it will hence be easier to sometimes work with the reverse time $r:=T-t$, looking backwards starting from $t=T$ and computing the value $r^*_j:=T-t_j^*$, starting from the smallest miner.

At $r=0$ (that is, $t=T$), we have already established that all miners will go for $\tx{2}{}$. At $r=1$, \eqref{eq:recursion} tells us that miner $j\in\qty{1,\ldots,m}$ will censor $\tx{1}{}$ as long as $f_1<\lambda_jf_2$. In other words, all miners with a hashrate smaller than or equal to $f_1/f_2$ will include $\tx{1}{}$ at $r=0$ (which is round $T-1$). Say there are $\ell$ miners with such a small hashrate, i.e., $\lambda_1\leq\ldots\lambda_\ell\leq f_1/f_2<\lambda_{\ell+1}$. Then for all $j=1,\ldots,\ell$, we have $r_j^*=0$ and so $t_j^*=T$. 

Let us now consider miner $j=\ell+1$. This miner will still censor at $r=1$, but how long will this miner keep censoring? We are looking for the smallest value of $t$ for which $g_{\ell+1}^{t+1}\qty((2)_{s=t+1}^{T})>f_1$. In other words, by Equation~\eqref{eq:censoring-gain}, we are looking for the smallest $t$ for which $\qty(\prod_{s=t+1}^{T-1}\lambda_C^s)\lambda_{\ell+1}f_2>f_1$.

We know that the hashrate $\sum_{k=1}^\ell\lambda_k$ is already trying to include $\tx{1}{}$, and that for all rounds larger than this smallest $t$, $j=\ell+1$ will be censoring. Hence, by Corollary~\ref{cor:start-censor}, we know that the portion $S_{\ell+1}:=\sum_{k=\ell+1}^m\lambda_k$ of hashrate will censor $\tx{1}{}$. That is, $\lambda_C^s=S_{\ell+1}=1-\lambda_0-\sum_{k=1}^\ell\lambda_k$ for $s=t,\ldots,T-1$. We are thus looking for $r_{\ell+1}^*$, the largest $r$ such that $S_{\ell+1}^{r-1}\lambda_{\ell+1}f_2>f_1$. It is easy to see that since $r\in\N_0$, $r_{\ell+1}^*-1=\floor{\rho_{\ell+1}}$, where $\rho_{\ell+1}$ satisfies $S_{\ell+1}^{\rho_{\ell+1}}\lambda_{\ell+1}f_2=f_1$, i.e., $\rho_{\ell+1}=\frac{\log\frac{f_1}{\lambda_{\ell+1}f_2}}{\log S_{\ell+1}}$.

In other words, we have $r_{\ell+1}^*=\ceil{\rho_{\ell+1}}$. Moving on, for the next miner $j=\ell+2$, we know that $r_j^*\geq r_{\ell+1}^*$. For $r>r_{\ell+1}^*$, miner $\ell+1$ is no longer censoring $\tx{1}{}$, hence from then on we have a portion $S_{\ell+2}$ censoring $\tx{1}{}$. We are thus looking for the largest $r$ such that $S_{\ell+2}^{r-r_{\ell+1}^*-1}S_{\ell+1}^{r_{\ell+1}^*}\lambda_{\ell+2}f_2>f_1$. Again, since $r\in\N_0$, we have $r_{\ell+2}^*-1=\floor{\rho_{\ell+2}}$, and so $r_{\ell+2}^*=\ceil{\rho_{\ell+2}}$, where $\rho_{\ell+2}$ is defined by $S_{\ell+2}^{\rho_{\ell+2}-\ceil{\rho_{\ell+1}}}S_{\ell+1}^{\ceil{\rho_{\ell+1}}}\lambda_{\ell+2}f_2=f_1$, giving $\rho_{\ell+2}=\ceil{\rho_{\ell+1}}+\frac{\log\frac{f_1}{\lambda_{\ell+2}f_2}-\ceil{\rho_{\ell+1}}\log S_{\ell+1}}{\log S_{\ell+2}}$. Repeating this argument inductively, one can see that more generally, for miner $j\in\qty{\ell+1,\ldots,m}$, we have $r_j^*=\ceil{\rho_j}$, where $\rho_j$ satisfies $S_j^{\rho_{j}-\ceil{\rho_{j-1}}}\prod_{i=\ell+1}^{j-1}\qty(S_i^{\ceil{\rho_{i}}-\ceil{\rho_{i-1}}})\lambda_{j}f_2=f_1$, where we set $\rho_\ell=0$. This gives us the following recursion for $\rho_j$:
\begin{footnotesize}
\begin{equation*}
\label{eq:rho-recursion}
    \rho_j=\ceil{\rho_{j-1}}+\frac{\log\frac{f_1}{\lambda_jf_2}-\sum_{i=\ell+1}^{j-1}\qty(\ceil{\rho_{i}}-\ceil{\rho_{i-1}})\log\qty(\sum_{k=i}^m\lambda_k)}{\log\sum_{k=j}^m\lambda_k}.
\end{equation*}
\end{footnotesize}
This identity allows us to compute recursively the values $r_j^*=\ceil{\rho_j}$ for every $j\in\qty{\ell+1,\ldots,m}$, for any given hashrate distribution $\lambdavec$ and fees $f_1,f_2$.

Now, given the values $(r_j^*)_{j=1}^m$ (keep in mind that $r_j^*=0$ for $j=1,\ldots,\ell$), we can write down the probability $P(C)$ that $\tx{2}{}$ will be included instead of $\tx{1}{}$, as a function of the timelock $T$. Clearly, if $T>r_m^*$, there will be at least one round in which all the miners will try to mine a block that includes $\tx{1}{}$, so $\tx{1}{}$ will be included with probability $1$. Trivially, since $f_2>f_1$, if $T=0$, all miners will go for $\tx{2}{}$ and so $\tx{1}{}$ will be included with probability $0$. Now, suppose that $T>r_{j-1}^*$ and $T\leq r_{j}^*$ for some $j\in\qty{1,\ldots,m}$ (we set $r_0^*=0$). Then the probability that $\tx{1}{}$ is censored for $T$ rounds is $\prod_{s=0}^{T-1}\lambda_C^s$. Since $\lambda_C^s=\sum_{k=j}^m\lambda_k$ for $s=r_{j-1}^*+1,\ldots,r_j^*$, we find the desired result 
\begin{footnotesize}
\begin{equation*}
    P(C)=\qty(\sum_{k=j}^m\lambda_k)^{T-r_{j-1}^*}\prod_{i=1}^{j-1}\qty(\qty(\sum_{k=i}^m\lambda_k)^{r_i^*-r_{i-1}^*}).
\end{equation*}
\end{footnotesize}

\subsection{Proof of Theorem~\ref{thm:2htlc-timelock}}
\label{subsec:thm:2htlc-timelock}
Suppose that $t^*_m(\overline{f^C_1},v_1;\lambdavec,T_1)+1< t^*_m(\overline{f^D_2},v_2;\lambdavec,T_2)$. Dave will deviate from the intended protocol behaviour in $\htlc_2$ by not sharing the secret with Charlie until $\tilde{t}:=t^*_m(\overline{f^C_1},v_1;\lambdavec,T_1)+1$. At time $\tilde{t}$, the Charlie-Dave channel will simply be updated off-chain (censoring would indeed fail with probability 1), giving Dave the same utility as if he would act according to the protocol. Charlie will want to claim $v_1$ in $\htlc_1$, but Dave will not respond, forcing Charlie to close the channel and post $\tx{1}{H},\tx{1}{C}$, the latter with a fee at most $\overline{f_1^C}$ (we can assume without loss of generality that $\tx{1}{H}$ has fee $0$). One can see that there must exist an $\varepsilon>0$ such that $t_m^*(\overline{f_1^C},v_1-\varepsilon;\lambdavec,T_1)=\tilde{t}$. Dave will post $\tx{1}{D}$ with the fee $v_1-\varepsilon$. With positive probability, $\tx{1}{D}$ will be included, in which case Dave gains $\varepsilon>0$. Dave's expected utility gain is positive and he will thus deviate from the IPB.

\subsection{Proof of Theorem~\ref{thm:wormhole}}
\label{subsec:thm:wormhole}
Assume without loss of generality that $t^s_3=0$. If everyone were to follow the protocol, we would have in $\htlc_3$ Dave sharing $s_3$ at $t=t^s_3=0$ with Charlie, who can update the channel, leading to the transaction triple $(x_3^P,0,0)$ being outputted. Charlie can now share $s_2$ at $t=0$ with Dave in $\htlc_2$, who will also update the channel, leading to $(x_2^P,0,0)$ being outputted. Similarly, in $\htlc_1$, Dave will share $s_1$ at $t=0$ with Alice, leading to $(x_1^P,0,0)$ being outputted. The blockchain game induced by the transaction triple set $\{(x_1^P,0,0),(x_2^P,0,0),(x_3^P,0,0)\}$ will trivially lead to an ordering which includes $x_1^P,x_2^P,x_3^P$. This implies a utility of $v^A_1$ for Alice, $v^C_2+v^C_3+v^2$ for Charlie, and $v^B_1+v^B_2+v^1+v^D_3+v^3$ for Dave.

However, Dave can perform the following deviation. Instead of sharing $s_3$ in $\htlc_3$, Dave does not share the secret at all, i.e., $\tau_D=\infty$. Since Charlie does not deviate, at $t=T_3$, Charlie will initiate a refund and Dave will accept, leading to the triple $(x^R_3,T_3,0)$. Since $\tau_{3,D}=\infty$, Dave will initiate a refund at $t=T_2$ in $\htlc_2$, and Charlie will accept, outputting  $(x^R_2,T_2,0)$. Finally, since Dave knows the secret $s_1$, Dave will actually share $s_1$ with Alice at $t=0$ in $\htlc_1$, outputting $(x^P_1,0,0)$. These three transaction triples will lead via the blockchain game again to a utility of $v^A_1$ for Alice, but to a utility of $v^C_2+v^C_3+v^3$ for Charlie, and of $v^B_1+v^1+v^B_2+v^2+v^D_3$ for Dave, who by deviating gains $v^B_1+v^1+v^B_2+v^2+v^D_3-(v^B_1+v^B_2+v^1+v^D_3+v^3)=v^2-v^3$. 

\subsection{Proof of Theorem~\ref{thm:crab}}
\label{subsec:thm:crab}
Realise that as a result of playing the sequence $(K_t)_{t=0}^T$ and the subsequent blockchain game, the resulting blockchain ordering will include the old commitment transactions and either a spending transaction or a punishment transaction. In the case of a spending transaction, Alice would receive an amount $v_A^o+c$, whereas in the case of a punishment transaction, she would receive $0$. If Alice would post the latest commitment transaction, she would receive $v_A^l+c$. Since Alice is rational, she will therefore not be willing to bribe the miners with more than $v_A^o-v_A^l$. If Bob already broadcasts in $K_0$ the punishment transaction, which rewards the miner immediately with an amount $c$, and if $c$ is high enough for the miners not to censor the punishment transaction, Alice will have no incentive to broadcast the spending transaction. 

This brings us in the setting of pre-BitVM Bitcoin timelock bribing, with the punishment transaction yielding an immediate reward of $c$ to the miners, and the spending transaction yielding a reward of at most $v_A^o-v_A^l$ after $T$ blocks. Hence, for a given hashrate distribution $\lambdavec$, Theorem~\ref{thm:tau} guarantees us that the punishment transaction will be included with probability $1$ if $T>\ceil{\rho_m(c,v_A^o-v_A^l;\lambdavec)}$. Since $v_A^o-v_A^l\leq v$, if $T>\ceil{\rho_m(c,v;\lambdavec)}$, the punishment transaction will always be included and so Alice has no incentive to broadcast an old commitment transaction. 

}

\end{document}